\documentclass{llncs}
\title{On the complexity of freezing automata networks of bounded pathwidth\thanks{Research partially supported by projects STIUC-AMSUD 22-STIC-02 (all authors), Fondecyt-ANID 1200006 (EG), FONDECYT-ANID 1230599 (PM), ANID FONDECYT Postdoctorado 3220205 (MRW)}} 

\author{Eric Goles\inst{1} \and Pedro Montealegre\inst{1} \and Martín Ríos-Wilson\inst{1} \and G.\ Theyssier\inst{2}}
\institute{Facultad de Ingeniería y Ciencias, Universidad Adolfo Ibáñez, Santiago \and I2M (Aix-Marseille University, CNRS)}

\usepackage{tikz,amsmath,amsfonts,diagbox}
\usetikzlibrary{arrows,decorations}

\newcommand\N{\mathbb{N}}
\newcommand{\cO}{\mathcal{O}}

\newcommand{\DLOG}{\textbf{DLOGSPACE}}

\newcommand\fop{FO$^+$}

\begin{document}

\maketitle

\begin{abstract}
 An automata network is a graph of entities, each holding a state from a finite set and evolving according to a local update rule which depends only on its neighbors in the network's graph. It is freezing if there is an order on the states such that the state evolution of any node is non-decreasing in any orbit. They are commonly used to model epidemic propagation, diffusion phenomena like bootstrap percolation or cristal growth. 
 
 Previous works have established that, under the hypothesis that the network graph is of bounded treewidth, many problems that can be captured by trace specifications at individual nodes admit efficient algorithms. In this paper we study the even more restricted case of a network of bounded pathwidth and show two hardness results that somehow illustrate the complexity of freezing dynamics under such a strong graph constraint. First, we show that the trace specification checking problem is NL-complete. Second, we show that deciding first order properties of the orbits augmented with a reachability predicate is NP-hard.
\end{abstract}
\section{Introduction}

Automata networks (AN) are finite dynamical systems that can be seen as the finite and non-uniform counterpart of cellular automata on arbitrary graphs. An automata network is \emph{freezing} if there is an order on the states such that the state evolution of any node is non-decreasing in any orbit. Several models that received a lot of attention in the literature are actually freezing automata networks, for instance: bootstrap percolation which has been studied on various graphs \cite{Amini_2014,Balogh_2012,Balogh_2005,Holroyd03}, epidemic \cite{Fuentes} or forest fire \cite{forestfire} propagation models, cristal growth models \cite{ulam,Gravner98} and various models of self-assembly tilings \cite{Winslow16}.

The freezing condition has strong implications on the computational complexity of these systems. For instance, following previous works on cellular automata \cite{dmtcs:9004,GolOlThey15}, it was established in \cite{papertw} that a large set of problems specified by traces at individual nodes are actually NC when considering freezing automata networks of bounded treewidth.
This result in particular captures the problem of nilpotency, a property which can be expressed in the language of orbits by: all orbits converge to the same fixed point. The nilpotency problem is typical of the computational complexity collapse when the freezing condition is combined by a condition on the structure of the network.

This paper aims at better understanding this complexity collapse by giving lower bounds for freezing automata networks on the simplest network structure: graphs of \emph{bounded pathwidth} (intuitively, that are structurally close to a line or a cycle).

 First, we consider regular trace properties (\textit{i.e.} regular expressions specifying allowed traces at each node) and show that the problem of existence of an orbit following the constraints is NL-complete (Theorem~\ref{theo:nlhard}). Note that this problem is similar to some well-studied problems in 1D cellular automata like cylinder-to-cylinder reachability which can also be expressed as a regular expression of traces \cite{DelvenneKB06,dmtcs:9004}. It is striking to compare the finite context with the NL upper bound above to the infinite context, where freezing cellular automata have actually an undecidable cylinder-to-cylinder reachability problem \cite{dmtcs:9004}.

 Second, we study another family of problems : properties defined by first order logic on configuration with equality, a predicate ${x\to y}$ meaning that $y$ can be reached from $x$ in one step, and a predicate ${x\to^+y}$ meaning that configuration $y$ can be reached from configuration $x$ in some number of steps. This logic denoted \fop also captures nilpotency by ${\exists! y,\forall x:x\to^+y}$. Our second main result is that, although nilpotency is {co-NL} (Corollary~\ref{cor:nilconl}), the model checking of \fop is NP-hard even for freezing automata networks defined on a line (Theorem \ref{theo:fop-hardness}).

\section{Definitions and notations}
Given a graph $G=(V,E)$ and a vertex $v$ we will call $N(v)$ the neighborhood of $v$ and $\delta_v$ to the degree of $v$.  In addition, we define the closed neighborhood of $v$ as the set $N[v] = N(v) \cup \{v\}$ and we use the following notation $\Delta(G) = \max \limits_{v \in V} \delta_v$  for the maximum degree of $G$. We will use the letter $n$ to denote the order of $G$, i.e. $n = |V|$.  Also, if $G$ is a graph and the set of vertices and edges is not specified we use the notation $V(G)$ and $E(G)$ for the set of vertices and the set of edges of $G$ respectively. In addition, we will assume that if $G =(V,E)$ is a graph then, there exists an ordering of the vertices in $V$ from $1$ to $n$.  During the rest of the text, every graph $G$ will be assumed to be connected and undirected.  We define a \textit{class} or a \textit{family} of graphs as a set $\mathcal{G}=\{G_n\}_{n \geq 1}$ such that $G_n = (V_n,E_n)$ is a graph and $|V_n| = n$. 

\emph{Non-deterministic freezing automata networks.}
Let $Q$ be a finite set that we will call an \textit{alphabet}. We define a non-deterministic automata network in the alphabet $Q$ as a tuple $(G=(V,E),\mathcal{F}=\{F_v : Q^{N(v)} \to \mathcal{P}(Q) | v \in V\}))$ where $\mathcal{P}(Q)$ is the power set of $Q$. To every non-deterministic automata network we can associate a non-deterministic dynamics given by the global function $F: Q^n \to \mathcal{P}(Q^n) $ defined by $F(x) = \{ y \in Q^n | y_v \in F_v(x),\forall v\}.$
\begin{definition}
	Given a a non-deterministic automata network $(G,\mathcal{F})$ we define an orbit of a configuration $x \in Q^n$ at time $t$ as a sequence $(x_s)_{0\leq s\leq t}$ such that $x_0 = x$ and $x_{s} \in F(x_{s-1}).$ In addition, we call the set of all possible orbits  at time $t$ for a configuration $x$ as $\mathcal{O} (x,t)$. Finally, we also define the set of all possible orbits at time $t$ as $\mathcal{O}(\mathcal{A},t
	) = \bigcup \limits_{x \in Q^n} \mathcal{O}(x,t)$
\end{definition}
We say that a non-deterministic automata network  $(G,\mathcal{F})$ defined in the alphabet $Q$ satisfies the \textit{freezing property} or simply that it is \textit{freezing} if there exists a partial order $\leq$ in $Q$ such that for every $t  \in \N$ and for every orbit $y = (x_s)_{0\leq s \leq t} \in \mathcal{O}(\mathcal{A},t)$ we have that $x^i_s \leq x^i_{s+1}$  for every $0 \leq s \leq t$ and for every $0\leq i \leq n.$

\emph{Path decompositions and pathwidth.} Let $G = (V,E)$ be a connected graph. A subgraph $P$ of $G$ is said to be a path if $V(P) = \{v_1,\hdots,v_k\}$ where every $v_i$ is different and $E(P) = \{v_1v_2 , v_2v_3\hdots, v_{k-1}v_k \}$. Now we present a graph parameter called \emph{pathwidth} which, generally speaking, indicates how similar a graph is to a path graph. More precisely, we have the following definition:
\begin{definition}
	Given a graph $G= (V,E)$ a path decomposition is pair $ \mathcal{D} = (P,\Lambda)$ such that $P$ is a path graph and $\Lambda$ is a family of subsets of nodes  $\Lambda = \{X_t \subseteq V | \text{ } t \in V(P)=\{1,\hdots,s\} \}$, called bags, such that:
	\begin{itemize}
		\item Every node in $G$ is in some $X_t$, i.e: $\bigcup \limits_{t \in V(P)} X_t = V,$
		\item For every $e=uv \in E$ there exists $t \in V(P)$ such that $u,v \in X_t,$
		\item For every $u,v,w \in V(P)$ if $1\leq u <v<w\leq s$ then, $X_u \cap X_w \subseteq X_v.$
	\end{itemize}
\end{definition}
We define the width of a path decompostion $\mathcal{D}$ as the amount $ \text{width} (\mathcal{D}) = \max  \limits_{t \in V(P)}|X_t|-1$.  Given a graph $G= (V,E)$, we define its pathwidth as the parameter  $\text{path}(G) = \min \limits_{\mathcal{D}} \text{width}(\mathcal{D})$. In other words, the pathwidth is the minimum width of a path decomposition of $G$. Note that, if $G$ is a connected graph such that $|E(G)| \geq 2$ then, $G$ is a path  if and only if $\text{path}(G) = 1$. 

It is known that a path decomposition of minimum width can be computed in $\DLOG$ \cite{kintali2012computing}.

\emph{Specification checking problem.}
Now, we introduce a decision problem called \emph{specification checking problem.} Roughly, this problem ask for the existence of an orbit in the automata network that verifies some trace constraints at each node. The information of allowed traces at each node is called a \emph{specification}: a specification of length $t$ is a map $\mathcal{E}_t: V \to \mathcal{P}(Q^t)$ such that, for every $v\in V$, the sequences in $\mathcal{E}_t(v)$ are non-decreasing (and thus respect the freezing condition). We say that $\mathcal{E}_t$ is satisfiable by $\mathcal{A}$ if there exists an orbit  $O \in \mathcal{O}(\mathcal{A},t)$ such that $O_v \in \mathcal{E}_t(v)$ for every $v \in V.$ 
We observe that the number of freezing traces of length $t$ is polynomial in $t$ so $\mathcal{E}_t$ can be represented in polynomial size in $V$ and $t$.

Also, in the absence of explicit mention, all the considered graphs will have bounded degree $\Delta$ by default, so a freezing automata network rule can be represented as the list of local update rules for each node which are maps of the form ${Q^\Delta\to \mathcal{P}(Q)}$ whose representation as transition table is of size ${O\bigl(|Q|^{\Delta+1}\bigr)}$  .
The specification checking problem (\textsf{SPEC}) introduced in \cite{papertw} asks whether a given freezing automata network satisfies a given specification.
If $\mathcal{E}_t$ is a satisfiable $t$-specification for some automata network $\mathcal{A}$ we write $\mathcal{A} \models \mathcal{E}_t.$ 


In \cite{papertw} it is shown that many well-known and well-studied decision problems related to the dynamics of automata networks are somehow related to \textsf{SPEC}. These problems are: the prediction problem, the predecessor problem, the nilpotency problem and the asynchronous reachability problem. Recall that nilpotency is the property that there is a configuration $x$ such that all orbits end up in $x$ and $x$ is a fixed point. Most of these problems are sub-problems of  \textsf{SPEC}. In the case of nilpotency, an efficient parallel Turing reduction can be constructed \cite{papertw}.   

In this paper, we focus on a variant of the specification problem were admissible traces are represented as regular expressions. More precisely,  a regular ${(Q,V)}$-specification is a map from $V$ to regular expressions over alphabet $Q$. We therefore consider the Regular Specification Checking Problem or simply \textsf{REGSPEC} which is the same as \textsf{SPEC} except that the specification must be a regular specification. It is interesting to observe that \textsf{REGSPEC} with fixed degree and fixed treewidth and with alphabet as unique parameter is ${W[2]}$-hard \cite{papertw}.

\paragraph{Recap of implicit hypothesis:} without explicit mention, our default object are non-deterministic freezing AN on connected and undirected graph of bounded degree $\Delta$.

\section{\textbf{NL}-completeness of REGSPEC problem}

In this section, we explore different results for the complexity of \textsf{REGSPEC} when the pathwidth of the underlying interaction graph is bounded. We start this section by showing that \textsf{REGSPEC} is in \textbf{NL}. This is a direct extension of the results on bounded treewidth in \cite{papertw} and the technique used in  \cite{dmtcs:9004} for the prediction problem in one dimensional freezing cellular automata.  Then, we show that the problem is actually \textbf{NL}-complete by showing a logspace reduction from $(s,t)$-connectivity.


\begin{theorem}\label{teo:NL}
The \textsf{REGSPEC} problem is in \textbf{NL} for bounded pathwidth (non-deterministic) freezing AN.
\end{theorem}
\begin{proof}
Let $t \in \N$ a time, $\mathcal{A} = (G, \mathcal{F})$ a non-deterministic automata network and $\mathcal{E}_t$ a $t$-specification. First note that if $G$ has bounded pathwidth, one can compute a path decomposition $\mathcal{P} = (X_{1}, \hdots, X_{p})$ in $\DLOG$ where $|X_i| \leq \text{pw}(G)$ for all ${1\leq i\leq p}$ (see \cite{kintali2012computing}). Now note that we can adapt the \textbf{NC} algorithm of \cite[Theorem 25]{papertw} to an \textbf{NL} algorithm in this particular context. First, observe that the dynamic programming lemma \cite[Lemma 19]{papertw}  is also valid in this case, but now, because the decomposition is a path, there is only one bag for each level. Then, observe that testing whether a trace in compact representation (as  explained earlier and presented in \cite{papertw}) belongs to some regular language can be done in $\DLOG$.  Then, the algorithm will reproduce the same procedure than the algorithm in \cite[Theorem 25]{papertw} , but, instead of parallelizing the information for the nodes in a bag storing it in different processors, it will handle this information non-deterministically.  More precisely, an algorithm can guess a trace for each bag $X_{l}$ from ${l=1}$ to ${l=p}$ while ensuring that each node (that can appear in various bags) has the same trace in all guesses: this can be done because, by definition of a path decomposition, a node appears in an interval of ${[1,p]}$. This is the major difference with \cite{papertw} that has to deal with tree decompositions.  Thus, \textsf{REGSPEC} problem is in \textbf{NL}.\qed
\end{proof}

The complement of the nilpotency problem can be reduced to instances of \textsf{REGSPEC} in such a way that we keep the strong complexity upper-bounds from the previous theorem.

\begin{corollary}\label{cor:nilconl}
  The nilpotency problem is in \textbf{co-NL} for bounded pathwidth freezing AN.
\end{corollary}

\begin{proof}
  For a freezing AN $F$ over alphabet $Q$, the property of \textbf{not} being nilpotent is equivalent to the existence of a pair of orbits that ends up in two fixed points that differ at some node. For any pair of states $q$ and $q'$ and some node $v$, denote by \textsf{NONIL}$(q,q',v)$ the problem of existence of two orbits in $F$ that end respectively in states $q$ and $q'$ at node $v$. \textsf{NONIL}$(q,q',v)$ is actually a \textsf{REGSPEC} problem for the AN $F\times F$ over alphabet $Q\times Q$ given by the following regular expression for trace at node $v$: ${(Q\times Q)^\ast(q,q')^+}$. Then, non-nilpotency can be expressed as the disjunction
  $$\bigvee_{v\in V}\bigvee_{q\neq q'}\mathsf{NONIL}(q,q',v).$$
  From this, we deduce a \textbf{NL} algorithm for non-nilpotency: choose non-deterministically one of the polynomially many instances of \textsf{NONIL} above and solve it in \textbf{NL} as an instance of the \textsf{REGSPEC} problem (Theorem~\ref{teo:NL}). We deduce that nilpotency is \textbf{co-NL}.\qed
\end{proof}

We now show that \textsf{REGSPEC} is \textbf{NL}-complete and thus, it is most likely that the previous algorithm is the best we can do, unless $\textbf{NL} = \DLOG.$ 

Now we introduce the main result of the section.

\begin{theorem}\label{theo:nlhard}
The Regular Specification Checking problem (\textsf{REGSPEC}) is \textbf{NL}-complete when restricted to bounded degree (non-deterministic) freezing AN with bounded pathwidth interaction graphs.
\end{theorem}

The proof proceeds by reduction from the problem STCON consisting in deciding, given a digraph $D$ and two nodes $s$ and $t$, whether there exists a path reaching $t$ from $s$.
The main idea is to construct a non-deterministic automata network $\mathcal{A}_{D} = (G_D, \mathcal{F}_D)$ defined over a two dimensional grid of size $k \times d$ where $d = n^{\cO(1)}$ and $k = \cO(1)$. Of course, since $k$ is constant, then $\mathcal{A}_{D}$ has bounded pathwidth. This automata network will non-deterministically guess a sequences of blocks (a structure representing edges in the interaction graph of $\mathcal{A}_D$, see Figure \ref{fig:blocktext} for more details). We call this part the \emph{selection phase}. Then, the next part of the proof consists in showing that $\mathcal{A}_{D} = (G_D, \mathcal{F}_D)$ is capable of deterministically verifying if an initial condition corresponds to a sequence of valid edges, i.e.  if it corresponds to a sequence of blocks and they actually represent edges in ${D}$. We call this phase a \emph{verification phase}. In order to perform this task,  we use a construction based on using signals that will collide at specific locations as a way to verify the distance between two given cells. In addition, it would be essential to save (as a constant layer) the information contained in the incidence matrix of $D$. Generally speaking, once $\mathcal{A}_D$ has verified that the sequence of blocks is valid, it will compare two subsequent blocks (which represent a pair of edges) in order to verify if they are incident. If in any part of its dynamics $\mathcal{A}_D$ locally detects some error (by the application of its local rule), it will spread an error state that will led the system to an attractor corresponding to a uniform configuration in which any cell will be in this particular error state. However, if the process runs flawless, then the system will reach an attractor in which all the cells are in a particular success state. We will code, by using a specification $\mathcal{E}_D$ (given in the input of \textsf{REGSPEC}), a specific requirement for the initial configuration (more precisely, we will ask the initial configuration to have the incidence matrix of $D$, markers and information about the nodes $(s,t)$) in order to allow $\mathcal{A}_D$ to have enough information to start the selection and verification process. In addition, we will code in this specification only the orbits that will reach this specific success state. By doing this, we will show that $\mathcal{A}_D \models \mathcal{E}_D$ if and only if there is a path between $s,t$ in $D$. Thus, the reduction will consist on constructing $(\mathcal{A}_D, \mathcal{E}_D)$ from $(D,s,t)$  in $\DLOG.$ 
\begin{figure}
\begin{tikzpicture}[x=0.75pt,y=0.75pt,yscale=-1,xscale=1]

\draw   (99,47+10) -- (534.5,47+10) -- (534.5,74) -- (99,74) -- cycle ;
\draw   (99,74) -- (534.5,74) -- (534.5,101-8) -- (99,101-8) -- cycle ;
\draw    (177.5,145.5) -- (250,145.02) ;
\draw [shift={(253,145)}, rotate = 539.62] [fill={rgb, 255:red, 0; green, 0; blue, 0 }  ][line width=0.08]  [draw opacity=0] (8.93,-4.29) -- (0,0) -- (8.93,4.29) -- cycle    ;
\draw  [fill={rgb, 255:red, 255; green, 255; blue, 255 }  ,fill opacity=1 ] (164,145.5) .. controls (164,138.04) and (170.04,132) .. (177.5,132) .. controls (184.96,132) and (191,138.04) .. (191,145.5) .. controls (191,152.96) and (184.96,159) .. (177.5,159) .. controls (170.04,159) and (164,152.96) .. (164,145.5) -- cycle ;
\draw    (266.5,144.5) -- (339,144.02) ;
\draw [shift={(342,144)}, rotate = 539.62] [fill={rgb, 255:red, 0; green, 0; blue, 0 }  ][line width=0.08]  [draw opacity=0] (8.93,-4.29) -- (0,0) -- (8.93,4.29) -- cycle    ;
\draw  [fill={rgb, 255:red, 255; green, 255; blue, 255 }  ,fill opacity=1 ] (253,144.5) .. controls (253,137.04) and (259.04,131) .. (266.5,131) .. controls (273.96,131) and (280,137.04) .. (280,144.5) .. controls (280,151.96) and (273.96,158) .. (266.5,158) .. controls (259.04,158) and (253,151.96) .. (253,144.5) -- cycle ;
\draw    (357.5,143.5) -- (430,143.02) ;
\draw [shift={(433,143)}, rotate = 539.62] [fill={rgb, 255:red, 0; green, 0; blue, 0 }  ][line width=0.08]  [draw opacity=0] (8.93,-4.29) -- (0,0) -- (8.93,4.29) -- cycle    ;
\draw  [fill={rgb, 255:red, 255; green, 255; blue, 255 }  ,fill opacity=1 ] (344,143.5) .. controls (344,136.04) and (350.04,130) .. (357.5,130) .. controls (364.96,130) and (371,136.04) .. (371,143.5) .. controls (371,150.96) and (364.96,157) .. (357.5,157) .. controls (350.04,157) and (344,150.96) .. (344,143.5) -- cycle ;
\draw  [fill={rgb, 255:red, 255; green, 255; blue, 255 }  ,fill opacity=1 ] (435,142.5) .. controls (435,135.04) and (441.04,129) .. (448.5,129) .. controls (455.96,129) and (462,135.04) .. (462,142.5) .. controls (462,149.96) and (455.96,156) .. (448.5,156) .. controls (441.04,156) and (435,149.96) .. (435,142.5) -- cycle ;

\draw (101,52.4+3) node [anchor=north west][inner sep=0.75pt]    {$ \begin{array}{l}
\#_{s} \ 1\ 1\ 1\ 1\ \#_{m} \ 1\ 1\ 1\ 1\ \#\ 0\ 0\ 0\ 0\ \#_{m} 0\ 0\ 0\ 0\ \#\ 0\ 0\ 0\ 0\ \#_{m}\ 0\ 0\ 0\ 0\ \#_{s}\\
\end{array}$};
\draw (101,52.4+20) node [anchor=north west][inner sep=0.75pt]    {$ \begin{array}{l}
\#_{s} \ h\ t\ 0\ 0\ \#_{m} \ 0\ 0\ t\ h\ \#\ h\ t\ 0\ 0\  \#_{m} \ 0\ 0\ t\ h\ \#\ h\ t\ 0\ 0\ \#_{m} \ 0\ 0\ t\ h\ \#_{s}
\end{array}$};
\draw (172,136.4) node [anchor=north west][inner sep=0.75pt]    {$v_1$};
\draw (261,135.4) node [anchor=north west][inner sep=0.75pt]    {$v_2$};
\draw (352,134.4) node [anchor=north west][inner sep=0.75pt]    {$v_3$};
\draw (443,133.4) node [anchor=north west][inner sep=0.75pt]    {$v_4$};
\draw (209,117.4) node [anchor=north west][inner sep=0.75pt]    {$e_{1}$};
\draw (301,115.4) node [anchor=north west][inner sep=0.75pt]    {$e_{2}$};
\draw (388,115.4) node [anchor=north west][inner sep=0.75pt]    {$e_{3}$};
\draw (306,21.4) node [anchor=north west][inner sep=0.75pt]    {$B( e_{1})$};
\draw (306,173.4) node [anchor=north west][inner sep=0.75pt]    {$D$};
\end{tikzpicture}
\caption{An example of a block for some graph $D$.}
\label{fig:blocktext}
\end{figure}
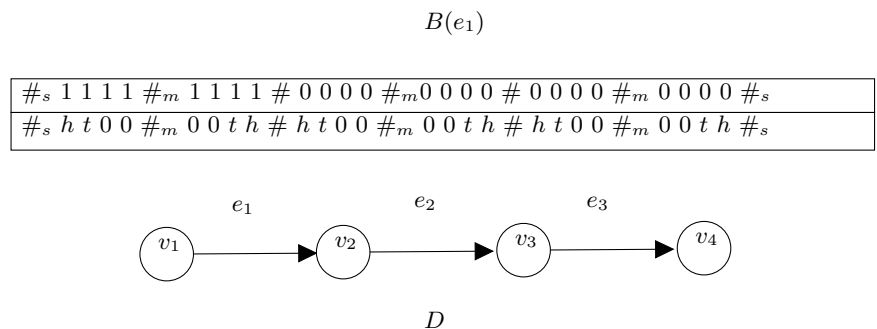

\subsection{Detailled construction and proof of Theorem~\ref{theo:nlhard}}

\begin{figure}
\centering
\tikzset{every picture/.style={line width=0.75pt}} 

\begin{tikzpicture}[x=0.75pt,y=0.75pt,yscale=-1,xscale=1]

\draw   (96,9) -- (545,9) -- (545,36.49) -- (96,36.49) -- cycle ;
\draw   (96,36.49) -- (545,36.49) -- (545,63.99) -- (96,63.99) -- cycle ;
\draw   (96,63.99) -- (545,63.99) -- (545,91.48) -- (96,91.48) -- cycle ;
\draw   (96,91.48) -- (545,91.48) -- (545,118.98) -- (96,118.98) -- cycle ;
\draw    (130.03+10,77.23) -- (181.4+10,77.23) ;
\draw [shift={(183.4,77.23)}, rotate = 180] [color={rgb, 255:red, 0; green, 0; blue, 0 }  ][line width=0.75]    (10.93,-3.29) .. controls (6.95,-1.4) and (3.31,-0.3) .. (0,0) .. controls (3.31,0.3) and (6.95,1.4) .. (10.93,3.29)   ;
\draw  [dash pattern={on 0.84pt off 2.51pt}]  (130.03+10,77.23) -- (130.03+10,54.82) ; 
\draw    (233.11+2,77.23) -- (182.4,77.23) ;
\draw [shift={(180.4,77.23)}, rotate = 360] [color={rgb, 255:red, 0; green, 0; blue, 0 }  ][line width=0.75]    (10.93,-3.29) .. controls (6.95,-1.4) and (3.31,-0.3) .. (0,0) .. controls (3.31,0.3) and (6.95,1.4) .. (10.93,3.29)   ;
\draw  [dash pattern={on 0.84pt off 2.51pt}]  (236.11,76.21) -- (236.11,53.8) ;
\draw  [dash pattern={on 0.84pt off 2.51pt}]  (181.4,87.43) -- (181.4,53.82) ;
\draw    (259.27,77.23) -- (310.65,77.23) ;
\draw [shift={(312.65,77.23)}, rotate = 180] [color={rgb, 255:red, 0; green, 0; blue, 0 }  ][line width=0.75]    (10.93,-3.29) .. controls (6.95,-1.4) and (3.31,-0.3) .. (0,0) .. controls (3.31,0.3) and (6.95,1.4) .. (10.93,3.29)   ;
\draw  [dash pattern={on 0.84pt off 2.51pt}]  (259.27,77.23) -- (259.27,54.82) ;
\draw    (373.35,77.23) -- (314.65,77.23) ;
\draw [shift={(312.65,77.23)}, rotate = 360] [color={rgb, 255:red, 0; green, 0; blue, 0 }  ][line width=0.75]    (10.93,-3.29) .. controls (6.95,-1.4) and (3.31,-0.3) .. (0,0) .. controls (3.31,0.3) and (6.95,1.4) .. (10.93,3.29)   ;
\draw  [dash pattern={on 0.84pt off 2.51pt}]  (373.35,76.21) -- (373.35,53.8) ;
\draw  [dash pattern={on 0.84pt off 2.51pt}]  (312.65,88.43) -- (312.65,54.82) ;

\draw    (403.71,77.23) -- (455.08,77.23) ;
\draw [shift={(457.08,77.23)}, rotate = 180] [color={rgb, 255:red, 0; green, 0; blue, 0 }  ][line width=0.75]    (10.93,-3.29) .. controls (6.95,-1.4) and (3.31,-0.3) .. (0,0) .. controls (3.31,0.3) and (6.95,1.4) .. (10.93,3.29)   ;
\draw  [dash pattern={on 0.84pt off 2.51pt}]  (403.71,77.23) -- (403.71,54.82) ;
\draw    (517.79-12,77.23) -- (459.08,77.23) ;
\draw [shift={(457.08,77.23)}, rotate = 360] [color={rgb, 255:red, 0; green, 0; blue, 0 }  ][line width=0.75]    (10.93,-3.29) .. controls (6.95,-1.4) and (3.31,-0.3) .. (0,0) .. controls (3.31,0.3) and (6.95,1.4) .. (10.93,3.29)   ;
\draw  [dash pattern={on 0.84pt off 2.51pt}]  (517.79-12,76.21) -- (517.79-12,53.8) ;
\draw  [dash pattern={on 0.84pt off 2.51pt}]  (457.08,88.43) -- (457.08,54.82) ;

\draw    (141.68+8,105.76) -- (179.34,105.76) ;
\draw [shift={(181.34,105.76)}, rotate = 180] [color={rgb, 255:red, 0; green, 0; blue, 0 }  ][line width=0.75]    (10.93,-3.29) .. controls (6.95,-1.4) and (3.31,-0.3) .. (0,0) .. controls (3.31,0.3) and (6.95,1.4) .. (10.93,3.29)   ;
\draw  [dash pattern={on 0.84pt off 2.51pt}]  (141.68+ 8,105.76) -- (141.68 + 8,53.82) ; 
\draw    (229.59,105.76) -- (183.34,105.76) ;
\draw [shift={(181.34,105.76)}, rotate = 360] [color={rgb, 255:red, 0; green, 0; blue, 0 }  ][line width=0.75]    (10.93,-3.29) .. controls (6.95,-1.4) and (3.31,-0.3) .. (0,0) .. controls (3.31,0.3) and (6.95,1.4) .. (10.93,3.29)   ;
\draw  [dash pattern={on 0.84pt off 2.51pt}]  (228.55,104.74) -- (227.5,54.84) ;
\draw  [dash pattern={on 0.84pt off 2.51pt}]  (181.34,117.98) -- (181.34,90.48) ;

\draw    (419.41,106.76) -- (457.07,106.76) ;
\draw [shift={(459.07,106.76)}, rotate = 180] [color={rgb, 255:red, 0; green, 0; blue, 0 }  ][line width=0.75]    (10.93,-3.29) .. controls (6.95,-1.4) and (3.31,-0.3) .. (0,0) .. controls (3.31,0.3) and (6.95,1.4) .. (10.93,3.29)   ;
\draw  [dash pattern={on 0.84pt off 2.51pt}]  (419.41,106.76) -- (418.36,54.82) ;
\draw    (507.32-10,106.76) -- (461.07,106.76) ;
\draw [shift={(459.07,106.76)}, rotate = 360] [color={rgb, 255:red, 0; green, 0; blue, 0 }  ][line width=0.75]    (10.93,-3.29) .. controls (6.95,-1.4) and (3.31,-0.3) .. (0,0) .. controls (3.31,0.3) and (6.95,1.4) .. (10.93,3.29)   ;
\draw  [dash pattern={on 0.84pt off 2.51pt}]  (506.28-10,105.74) -- (505.23-10,55.84) ;
\draw  [dash pattern={on 0.84pt off 2.51pt}]  (459.07-2,118.98) -- (459.07-2,91.48) ;

\draw    (276,107) -- (310.54,106.77) ;
\draw [shift={(312.54,106.76)}, rotate = 539.62] [color={rgb, 255:red, 0; green, 0; blue, 0 }  ][line width=0.75]    (10.93,-3.29) .. controls (6.95,-1.4) and (3.31,-0.3) .. (0,0) .. controls (3.31,0.3) and (6.95,1.4) .. (10.93,3.29)   ;
\draw  [dash pattern={on 0.84pt off 2.51pt}]  (273.79,106.76) -- (272.74,54.82) ;
\draw    (364.98,106.76) -- (314.54,106.76) ;
\draw [shift={(312.54,106.76)}, rotate = 360] [color={rgb, 255:red, 0; green, 0; blue, 0 }  ][line width=0.75]    (10.93,-3.29) .. controls (6.95,-1.4) and (3.31,-0.3) .. (0,0) .. controls (3.31,0.3) and (6.95,1.4) .. (10.93,3.29)   ;
\draw  [dash pattern={on 0.84pt off 2.51pt}]  (363.93,105.74) -- (362.89,55.84) ;
\draw  [dash pattern={on 0.84pt off 2.51pt}]  (312.54,118.98) -- (312.54,91.48) ;
\draw   (96,118.98) -- (545,118.98) -- (545,146.47) -- (96,146.47) -- cycle ;
\draw  [dash pattern={on 0.84pt off 2.51pt}]  (236.43,137.3) -- (236.11,77.23) ;
\draw  [dash pattern={on 0.84pt off 2.51pt}]  (260.32,137.3) -- (259.27,77.23) ;
\draw  [dash pattern={on 0.84pt off 2.51pt}]  (250,172) -- (248.67,53.82) ;
\draw    (374.72,137.3) -- (384.85,137.3) ;
\draw [shift={(386.85,137.3)}, rotate = 180] [color={rgb, 255:red, 0; green, 0; blue, 0 }  ][line width=0.75]    (10.93,-3.29) .. controls (6.95,-1.4) and (3.31,-0.3) .. (0,0) .. controls (3.31,0.3) and (6.95,1.4) .. (10.93,3.29)   ;
\draw  [dash pattern={on 0.84pt off 2.51pt}]  (374.72,137.3) -- (374.4,77.23) ;
\draw    (401.61,137.3) -- (388.85,137.3) ;
\draw [shift={(386.85,137.3)}, rotate = 360] [color={rgb, 255:red, 0; green, 0; blue, 0 }  ][line width=0.75]    (10.93,-3.29) .. controls (6.95,-1.4) and (3.31,-0.3) .. (0,0) .. controls (3.31,0.3) and (6.95,1.4) .. (10.93,3.29)   ;
\draw  [dash pattern={on 0.84pt off 2.51pt}]  (401.61+3,137.3) -- (401.61+3,77.23) ;
\draw  [dash pattern={on 0.84pt off 2.51pt}]  (386.96,147.49) -- (386.96,54.82) ;
\draw   (96,146.47) -- (545,146.47) -- (545,173.96) -- (96,173.96) -- cycle ;
\draw    (229.24,164.8) -- (248.11,164.8) ;
\draw [shift={(250.11,164.8)}, rotate = 180] [color={rgb, 255:red, 0; green, 0; blue, 0 }  ][line width=0.75]    (10.93,-3.29) .. controls (6.95,-1.4) and (3.31,-0.3) .. (0,0) .. controls (3.31,0.3) and (6.95,1.4) .. (10.93,3.29)   ;
\draw  [dash pattern={on 0.84pt off 2.51pt}]  (229.24,164.8) -- (228.64,104.72) ;
\draw    (275.76,164.84) -- (250.11,164.8) ;
\draw [shift={(248.11,164.8)}, rotate = 360.08000000000004] [color={rgb, 255:red, 0; green, 0; blue, 0 }  ][line width=0.75]    (10.93,-3.29) .. controls (6.95,-1.4) and (3.31,-0.3) .. (0,0) .. controls (3.31,0.3) and (6.95,1.4) .. (10.93,3.29)   ;
\draw  [dash pattern={on 0.84pt off 2.51pt}]  (275.76,166.84) -- (273.79,106.76) ;
\draw    (364.58,165.26) -- (386.84,165.26) ;
\draw [shift={(388.84,165.26)}, rotate = 180] [color={rgb, 255:red, 0; green, 0; blue, 0 }  ][line width=0.75]    (10.93,-3.29) .. controls (6.95,-1.4) and (3.31,-0.3) .. (0,0) .. controls (3.31,0.3) and (6.95,1.4) .. (10.93,3.29)   ;
\draw  [dash pattern={on 0.84pt off 2.51pt}]  (364.58,165.26) -- (363.93,101.88) ;
\draw    (421.5,165.26) -- (388.75,165.26) ;
\draw [shift={(386.75,165.26)}, rotate = 360] [color={rgb, 255:red, 0; green, 0; blue, 0 }  ][line width=0.75]    (10.93,-3.29) .. controls (6.95,-1.4) and (3.31,-0.3) .. (0,0) .. controls (3.31,0.3) and (6.95,1.4) .. (10.93,3.29)   ;
\draw  [dash pattern={on 0.84pt off 2.51pt}]  (420.45,165.82) -- (419.41,106.76) ;
\draw  [dash pattern={on 0.84pt off 2.51pt}]  (386.96,176) -- (386.96,78.24) ;
\draw    (236.43,137.3) -- (247,137.05) ;
\draw [shift={(249,137)}, rotate = 538.61] [color={rgb, 255:red, 0; green, 0; blue, 0 }  ][line width=0.75]    (10.93,-3.29) .. controls (6.95,-1.4) and (3.31,-0.3) .. (0,0) .. controls (3.31,0.3) and (6.95,1.4) .. (10.93,3.29)   ;
\draw    (260.32,137.3) -- (252,137.06) ;
\draw [shift={(250,137)}, rotate = 361.69] [color={rgb, 255:red, 0; green, 0; blue, 0 }  ][line width=0.75]    (10.93,-3.29) .. controls (6.95,-1.4) and (3.31,-0.3) .. (0,0) .. controls (3.31,0.3) and (6.95,1.4) .. (10.93,3.29)   ;

\draw (106+10,13.4) node [anchor=north west][inner sep=0.75pt]    {$ \begin{array}{l}
\#_{s} \ 1\ 1\ 1\ 1\ \#_{m} \ 1\ 1\ 1\ 1 \#\ 0\ 0\ 0\ 0\ \#_{m} 0\ 0\ \ \ \ \ \ \  0\ 0\ \#\ \ 0\ 0\ 0\ 0\ \#_{m} 0\ 0\ 0\ 0 \#_{s}\\
\end{array}$};
\draw (106+10,13.4+25) node [anchor=north west][inner sep=0.75pt]    {$ \begin{array}{l}
\#_{s} \ h\ t\ 0\ 0\ \#_{m} \ 0\ 0\ t\ h\ \#\ h\ t\ 0\ 0\  \#_{m} \ 0\ 0\ \ \ \ \ \ t\ h\ \#\ \  h\ t\ 0\ 0\ \#_{m} \ 0\ 0\ t\ h \#_{s}
\end{array}$};
\end{tikzpicture}
\caption{Example of the verification dynamics for a periodic pattern. In this case, the pattern is given by the second row of a block representing the edge $(1,2)$ in some graph.}
\label{fig:verdym}
\end{figure}
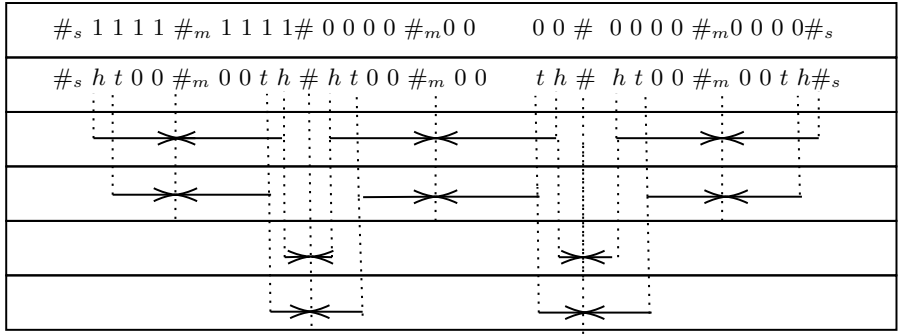

\begin{figure}
\centering

\tikzset{every picture/.style={line width=0.75pt}} 

\begin{tikzpicture}[x=0.75pt,y=0.75pt,yscale=-1,xscale=1]

\draw  [fill={rgb, 255:red, 255; green, 255; blue, 255 }  ,fill opacity=1 ]    (106.76, 0) circle [x radius= 12, y radius= 12] ;
\draw (106.76+6,-5) node [anchor=north east][inner sep=0.75pt]    {A};

\draw   (101,25) -- (537.43,25) -- (537.43,52.63) -- (101,52.63) -- cycle ;
\draw   (101,52.63) -- (537.43,52.63) -- (537.43,80.25) -- (101,80.25) -- cycle ;
\draw   (101,80.25) -- (537.43,80.25) -- (537.43,107.88) -- (101,107.88) -- cycle ;
\draw    (112.7,65.93) -- (175.26,65.93) ;
\draw    (254.51,65.93) -- (311.59,65.93) ;
\draw    (175.26,65.93) -- (251.51,65.93) ;
\draw [shift={(254.51,65.93)}, rotate = 180] [fill={rgb, 255:red, 0; green, 0; blue, 0 }  ][line width=0.08]  [draw opacity=0] (8.93,-4.29) -- (0,0) -- (8.93,4.29) -- cycle    ;
\draw    (112.19,95.6) -- (175.16,95.6) ;
\draw  [dash pattern={on 0.84pt off 2.51pt}]  (112.7,134.99) -- (111.17,43.42) ;
\draw    (255.51,95.6) -- (311.48,95.6) ;
\draw  [dash pattern={on 0.84pt off 2.51pt}]  (256.5,81.5) -- (255.55,43.42) ;
\draw    (175.16,95.6) -- (252.51,95.6) ;
\draw [shift={(255.51,95.6)}, rotate = 180] [fill={rgb, 255:red, 0; green, 0; blue, 0 }  ][line width=0.08]  [draw opacity=0] (8.93,-4.29) -- (0,0) -- (8.93,4.29) -- cycle    ;
\draw    (394.75,65.93) -- (451.47,65.93) ;
\draw    (311.59,65.93) -- (384.75,65.93) ;
\draw [shift={(387.75,65.93)}, rotate = 180] [fill={rgb, 255:red, 0; green, 0; blue, 0 }  ][line width=0.08]  [draw opacity=0] (8.93,-4.29) -- (0,0) -- (8.93,4.29) -- cycle    ;
\draw    (395.75,95.6) -- (451.47,95.6) ;
\draw    (311.48,95.6) -- (387.64,95.6) ;
\draw [shift={(390.64,95.6)}, rotate = 180] [fill={rgb, 255:red, 0; green, 0; blue, 0 }  ][line width=0.08]  [draw opacity=0] (8.93,-4.29) -- (0,0) -- (8.93,4.29) -- cycle    ;
\draw    (451.47,65.93) -- (521.38,65.93) ;
\draw [shift={(524.38,65.93)}, rotate = 180] [fill={rgb, 255:red, 0; green, 0; blue, 0 }  ][line width=0.08]  [draw opacity=0] (8.93,-4.29) -- (0,0) -- (8.93,4.29) -- cycle    ;
\draw    (451.47,95.6) -- (521.38,95.6) ;
\draw [shift={(524.38,95.6)}, rotate = 180] [fill={rgb, 255:red, 0; green, 0; blue, 0 }  ][line width=0.08]  [draw opacity=0] (8.93,-4.29) -- (0,0) -- (8.93,4.29) -- cycle    ;
\draw   (101,107.88) -- (537.43,107.88) -- (537.43,135.5) -- (101,135.5) -- cycle ;
\draw    (114.19,120.6) -- (257,120.5) ;
\draw [shift={(259,120.5)}, rotate = 539.96] [fill={rgb, 255:red, 0; green, 0; blue, 0 }  ][line width=0.08]  [draw opacity=0] (12,-3) -- (0,0) -- (12,3) -- cycle    ;
\draw    (259,120.5) -- (521.15,120.5) ;
\draw [shift={(523.5,120.5)}, rotate = 0] [color={rgb, 255:red, 0; green, 0; blue, 0 }  ][line width=0.75]      (0, 0) circle [x radius= 3.35, y radius= 3.35]   ;
\draw    (257.45,107.58) .. controls (255.72,105.98) and (255.66,104.32) .. (257.27,102.59) .. controls (258.88,100.86) and (258.82,99.2) .. (257.09,97.59) .. controls (255.36,95.98) and (255.3,94.32) .. (256.91,92.59) .. controls (258.51,90.86) and (258.45,89.2) .. (256.72,87.6) .. controls (254.99,85.99) and (254.93,84.33) .. (256.54,82.6) -- (256.5,81.5) -- (256.5,81.5) ;
\draw    (390.22,78.97) .. controls (388.49,77.36) and (388.43,75.7) .. (390.04,73.97) .. controls (391.65,72.24) and (391.59,70.58) .. (389.86,68.97) .. controls (388.13,67.37) and (388.07,65.71) .. (389.67,63.98) .. controls (391.28,62.25) and (391.22,60.59) .. (389.49,58.98) .. controls (387.76,57.37) and (387.7,55.71) .. (389.31,53.98) -- (389.27,52.88) -- (389.27,52.88) -- (389.67,52.88) -- (389.27-1,52.88-2-2) -- (389.67,52.88-2-2-2) -- (389.27-1,52.88-2-2-2-2)  ;
\draw  [dash pattern={on 0.84pt off 2.51pt}]  (390.5,106.5) -- (390.22,78.97) ;
\draw  [dash pattern={on 0.84pt off 2.51pt}]  (258,136.5) -- (257.45,107.58) ;
\draw  [dash pattern={on 3.75pt off 3pt on 7.5pt off 1.5pt}][line width=0.75]  (241.5,58.5) -- (271.5,58.5) -- (271.5,115.27) -- (241.5,115.27) -- cycle ;

\draw (106.76,28.68) node [anchor=north west][inner sep=0.75pt]    {$\#_{s} \ 1\ 1\ 1\ 1\ \#_{m} \ 1\ 1\ 1\ 1\  \  \ \ \ \#\ 0\ 0\ 0\ 0\ \#_{m} 0\ 0\ 0\ 0\  \ \ \ \ \#\ 0\ 0\ 0\ 0\ \#_{m} 0\ 0\ 0\ 0\ \ \ \#_{s}$};
\draw (517.51,146.4) node [anchor=north west][inner sep=0.75pt]    {$\text{OK} $};
\draw    (100,146.4+20) -- (517.51+20,146.4+20) ;
\draw (517.51,146.4+25) node [anchor=north west][inner sep=0.75pt]    {~};

\end{tikzpicture}


\begin{tikzpicture}[x=0.75pt,y=0.75pt,yscale=-1,xscale=1]
\draw  [fill={rgb, 255:red, 255; green, 255; blue, 255 }  ,fill opacity=1 ]    (106.76, 0) circle [x radius= 12, y radius= 12] ;
\draw (106.76+6,-5) node [anchor=north east][inner sep=0.75pt]    {B};
\draw   (101,25) -- (537.43,25) -- (537.43,52.63) -- (101,52.63) -- cycle ;
\draw   (101,52.63) -- (537.43,52.63) -- (537.43,80.25) -- (101,80.25) -- cycle ;
\draw   (101,80.25) -- (537.43,80.25) -- (537.43,107.88) -- (101,107.88) -- cycle ;
\draw    (112.7,65.93) -- (175.26,65.93) ;
\draw    (254.51,65.93) -- (311.59,65.93) ;
\draw    (175.26,65.93) -- (251.51,65.93) ;
\draw [shift={(254.51,65.93)}, rotate = 180] [fill={rgb, 255:red, 0; green, 0; blue, 0 }  ][line width=0.08]  [draw opacity=0] (8.93,-4.29) -- (0,0) -- (8.93,4.29) -- cycle    ;
\draw    (112.19,95.6) -- (175.16,95.6) ;
\draw  [dash pattern={on 0.84pt off 2.51pt}]  (112.7,134.99) -- (111.17,43.42) ;
\draw    (255.51,95.6) -- (311.48,95.6) ;
\draw    (256.5,81.5) .. controls (254.79,79.87) and (254.75,78.21) .. (256.37,76.5) .. controls (258,74.79) and (257.96,73.13) .. (256.25,71.5) .. controls (254.54,69.87) and (254.5,68.21) .. (256.12,66.5) .. controls (257.75,64.79) and (257.71,63.13) .. (256,61.51) .. controls (254.29,59.88) and (254.25,58.22) .. (255.87,56.51) .. controls (257.5,54.8) and (257.46,53.14) .. (255.75,51.51) .. controls (254.04,49.88) and (254,48.22) .. (255.62,46.51) -- (255.55,43.42) -- (255.55,43.42) ;
\draw    (175.16,95.6) -- (252.51,95.6) ;
\draw [shift={(255.51,95.6)}, rotate = 180] [fill={rgb, 255:red, 0; green, 0; blue, 0 }  ][line width=0.08]  [draw opacity=0] (8.93,-4.29) -- (0,0) -- (8.93,4.29) -- cycle    ;
\draw    (394.75,65.93) -- (451.47,65.93) ;
\draw    (311.59,65.93) -- (384.75,65.93) ;
\draw [shift={(387.75,65.93)}, rotate = 180] [fill={rgb, 255:red, 0; green, 0; blue, 0 }  ][line width=0.08]  [draw opacity=0] (8.93,-4.29) -- (0,0) -- (8.93,4.29) -- cycle    ;
\draw    (395.75,95.6) -- (451.47,95.6) ;
\draw    (311.48,95.6) -- (387.64,95.6) ;
\draw [shift={(390.64,95.6)}, rotate = 180] [fill={rgb, 255:red, 0; green, 0; blue, 0 }  ][line width=0.08]  [draw opacity=0] (8.93,-4.29) -- (0,0) -- (8.93,4.29) -- cycle    ;
\draw    (451.47,65.93) -- (521.38,65.93) ;
\draw [shift={(524.38,65.93)}, rotate = 180] [fill={rgb, 255:red, 0; green, 0; blue, 0 }  ][line width=0.08]  [draw opacity=0] (8.93,-4.29) -- (0,0) -- (8.93,4.29) -- cycle    ;
\draw    (451.47,95.6) -- (521.38,95.6) ;
\draw [shift={(524.38,95.6)}, rotate = 180] [fill={rgb, 255:red, 0; green, 0; blue, 0 }  ][line width=0.08]  [draw opacity=0] (8.93,-4.29) -- (0,0) -- (8.93,4.29) -- cycle    ;
\draw   (101,107.88) -- (537.43,107.88) -- (537.43,135.5) -- (101,135.5) -- cycle ;
\draw    (114.19,120.6) -- (257,120.5) ;
\draw [shift={(259,120.5)}, rotate = 539.96] [fill={rgb, 255:red, 0; green, 0; blue, 0 }  ][line width=0.08]  [draw opacity=0] (12,-3) -- (0,0) -- (12,3) -- cycle    ;
\draw    (259,120.5) -- (523.5,120.5) ;
\draw [shift={(523.5,120.5)}, rotate = 0] [color={rgb, 255:red, 0; green, 0; blue, 0 }  ][fill={rgb, 255:red, 0; green, 0; blue, 0 }  ][line width=0.75]      (0, 0) circle [x radius= 3.35, y radius= 3.35]   ;
\draw    (257.45,107.58) .. controls (255.72,105.98) and (255.66,104.32) .. (257.27,102.59) .. controls (258.88,100.86) and (258.82,99.2) .. (257.09,97.59) .. controls (255.36,95.98) and (255.3,94.32) .. (256.91,92.59) .. controls (258.51,90.86) and (258.45,89.2) .. (256.72,87.6) .. controls (254.99,85.99) and (254.93,84.33) .. (256.54,82.6) -- (256.5,81.5) -- (256.5,81.5) ;
\draw    (390.22,78.97) .. controls (388.49,77.36) and (388.43,75.7) .. (390.04,73.97) .. controls (391.65,72.24) and (391.59,70.58) .. (389.86,68.97) .. controls (388.13,67.37) and (388.07,65.71) .. (389.67,63.98) .. controls (391.28,62.25) and (391.22,60.59) .. (389.49,58.98) .. controls (387.76,57.37) and (387.7,55.71) .. (389.31,53.98) -- (389.27,52.88) -- (389.27,52.88) ;
\draw  [dash pattern={on 0.84pt off 2.51pt}]  (390.5,106.5) -- (390.22,78.97) ;
\draw  [dash pattern={on 0.84pt off 2.51pt}]  (258,136.5) -- (257.45,107.58) ;
\draw  [dash pattern={on 3.75pt off 3pt on 7.5pt off 1.5pt}][line width=0.75]  (241.5,58.5) -- (271.5,58.5) -- (271.5,115.27) -- (241.5,115.27) -- cycle ;

\draw (106.76,28.68) node [anchor=north west][inner sep=0.75pt]    {$\#_{s} \ 1\ 1\ 1\ 1\ \#_{m} \ 1\ 0\ 0\ 1\  \ \ \ \ \#\ 0\ 0\ 0\ 0\ \#_{m} 0\ 0\ 0\ 0\  \ \ \ \  \#\ 0\ 0\ 0\ 0\ \#_{m} 0\ 0\ 0\ 0\ \ \ \#_{s}$};
\draw (250.51,139) node [anchor=north west][inner sep=0.75pt]    {$\times $};
\draw    (100,146.4+20) -- (517.51+20,146.4+20) ;
\draw (517.51,146.4+25) node [anchor=north west][inner sep=0.75pt]    {~};
\end{tikzpicture}

\tikzset{every picture/.style={line width=0.75pt}} 

\begin{tikzpicture}[x=0.75pt,y=0.75pt,yscale=-1,xscale=1]
\draw   [fill={rgb, 255:red, 255; green, 255; blue, 255 }  ,fill opacity=1 ]   (106.76, 0) circle [x radius= 12, y radius= 12] ;
\draw (106.76+6,-5) node [anchor=north east][inner sep=0.75pt]    {C};
\draw   (101,25) -- (537.43,25) -- (537.43,52.63) -- (101,52.63) -- cycle ;
\draw   (101,52.63) -- (537.43,52.63) -- (537.43,80.25) -- (101,80.25) -- cycle ;
\draw   (101,80.25) -- (537.43,80.25) -- (537.43,107.88) -- (101,107.88) -- cycle ;
\draw    (112.7,65.93) -- (175.26,65.93) ;
\draw    (254.51,65.93) -- (311.59,65.93) ;
\draw    (175.26,65.93) -- (251.51,65.93) ;
\draw [shift={(254.51,65.93)}, rotate = 180] [fill={rgb, 255:red, 0; green, 0; blue, 0 }  ][line width=0.08]  [draw opacity=0] (8.93,-4.29) -- (0,0) -- (8.93,4.29) -- cycle    ;
\draw    (112.19,95.6) -- (175.16,95.6) ;
\draw  [dash pattern={on 0.84pt off 2.51pt}]  (112.7,134.99) -- (111.17,43.42) ;
\draw    (255.51,95.6) -- (311.48,95.6) ;
\draw    (256.5,81.5) .. controls (254.79,79.87) and (254.75,78.21) .. (256.37,76.5) .. controls (258,74.79) and (257.96,73.13) .. (256.25,71.5) .. controls (254.54,69.87) and (254.5,68.21) .. (256.12,66.5) .. controls (257.75,64.79) and (257.71,63.13) .. (256,61.51) .. controls (254.29,59.88) and (254.25,58.22) .. (255.87,56.51) .. controls (257.5,54.8) and (257.46,53.14) .. (255.75,51.51) .. controls (254.04,49.88) and (254,48.22) .. (255.62,46.51) -- (255.55,43.42) -- (255.55,43.42) ;
\draw    (175.16,95.6) -- (252.51,95.6) ;
\draw [shift={(255.51,95.6)}, rotate = 180] [fill={rgb, 255:red, 0; green, 0; blue, 0 }  ][line width=0.08]  [draw opacity=0] (8.93,-4.29) -- (0,0) -- (8.93,4.29) -- cycle    ;
\draw    (394.75,65.93) -- (451.47,65.93) ;
\draw    (311.59,65.93) -- (384.75,65.93) ;
\draw [shift={(387.75,65.93)}, rotate = 180] [fill={rgb, 255:red, 0; green, 0; blue, 0 }  ][line width=0.08]  [draw opacity=0] (8.93,-4.29) -- (0,0) -- (8.93,4.29) -- cycle    ;
\draw    (395.75,95.6) -- (451.47,95.6) ;
\draw    (311.48,95.6) -- (387.64,95.6) ;
\draw [shift={(390.64,95.6)}, rotate = 180] [fill={rgb, 255:red, 0; green, 0; blue, 0 }  ][line width=0.08]  [draw opacity=0] (8.93,-4.29) -- (0,0) -- (8.93,4.29) -- cycle    ;
\draw    (451.47,65.93) -- (521.38,65.93) ;
\draw [shift={(524.38,65.93)}, rotate = 180] [fill={rgb, 255:red, 0; green, 0; blue, 0 }  ][line width=0.08]  [draw opacity=0] (8.93,-4.29) -- (0,0) -- (8.93,4.29) -- cycle    ;
\draw    (451.47,95.6) -- (521.38,95.6) ;
\draw [shift={(524.38,95.6)}, rotate = 180] [fill={rgb, 255:red, 0; green, 0; blue, 0 }  ][line width=0.08]  [draw opacity=0] (8.93,-4.29) -- (0,0) -- (8.93,4.29) -- cycle    ;
\draw   (101,107.88) -- (537.43,107.88) -- (537.43,135.5) -- (101,135.5) -- cycle ;
\draw    (114.19,120.6) -- (257,120.5) ;
\draw [shift={(259,120.5)}, rotate = 539.96] [fill={rgb, 255:red, 0; green, 0; blue, 0 }  ][line width=0.08]  [draw opacity=0] (12,-3) -- (0,0) -- (12,3) -- cycle    ;
\draw    (259,120.5) -- (523.5,120.5) ;
\draw [shift={(523.5,120.5)}, rotate = 0] [color={rgb, 255:red, 0; green, 0; blue, 0 }  ][fill={rgb, 255:red, 0; green, 0; blue, 0 }  ][line width=0.75]      (0, 0) circle [x radius= 3.35, y radius= 3.35]   ;
\draw  [dash pattern={on 0.84pt off 2.51pt}]  (257.45,107.58) -- (256.5,81.5) ;
\draw    (390.22,78.97) .. controls (388.49,77.36) and (388.43,75.7) .. (390.04,73.97) .. controls (391.65,72.24) and (391.59,70.58) .. (389.86,68.97) .. controls (388.13,67.37) and (388.07,65.71) .. (389.67,63.98) .. controls (391.28,62.25) and (391.22,60.59) .. (389.49,58.98) .. controls (387.76,57.37) and (387.7,55.71) .. (389.31,53.98) -- (389.27,52.88) -- (389.27,52.88) ;
\draw  [dash pattern={on 0.84pt off 2.51pt}]  (390.5,106.5) -- (390.22,78.97) ;
\draw  [dash pattern={on 0.84pt off 2.51pt}]  (258,136.5) -- (257.45,107.58) ;

\draw (106.76,28.68) node [anchor=north west][inner sep=0.75pt]    {$\#_{s} \ 0\ 0\ 0\ 0\ \#_{m} \ 0\ 0\ 0\ 0\  \  \  \ \ \#\ 0\ 0\ 0\ 0\ \#_{m} 0\ 0\ 0\ 0\  \ \ \ \ \ \#\ 0\ 0\ 0\ 0\ \#_{m} 0\ 0\ 0\ 0\ \ \ \ \#_{s}$};
\draw (517.51,140.4) node [anchor=north west][inner sep=0.75pt]    {$\times $};
\end{tikzpicture}
\caption{Example of the verification dynamics for a marker. (Upper panel) A successful verification of a marker. If exactly one zone has only cells in state $1$ an acceptance state will be reached. (Middle panel) An error in the verification raised by a zone in which cells in state $0$ and $1$ were identified.  In this case, an error state is propagated. (Lower panel) An error in the verification raised by the signal only reading cells in state $0$.  In this case, an error state is propagated. }
\label{fig:secondrowver}
\end{figure}
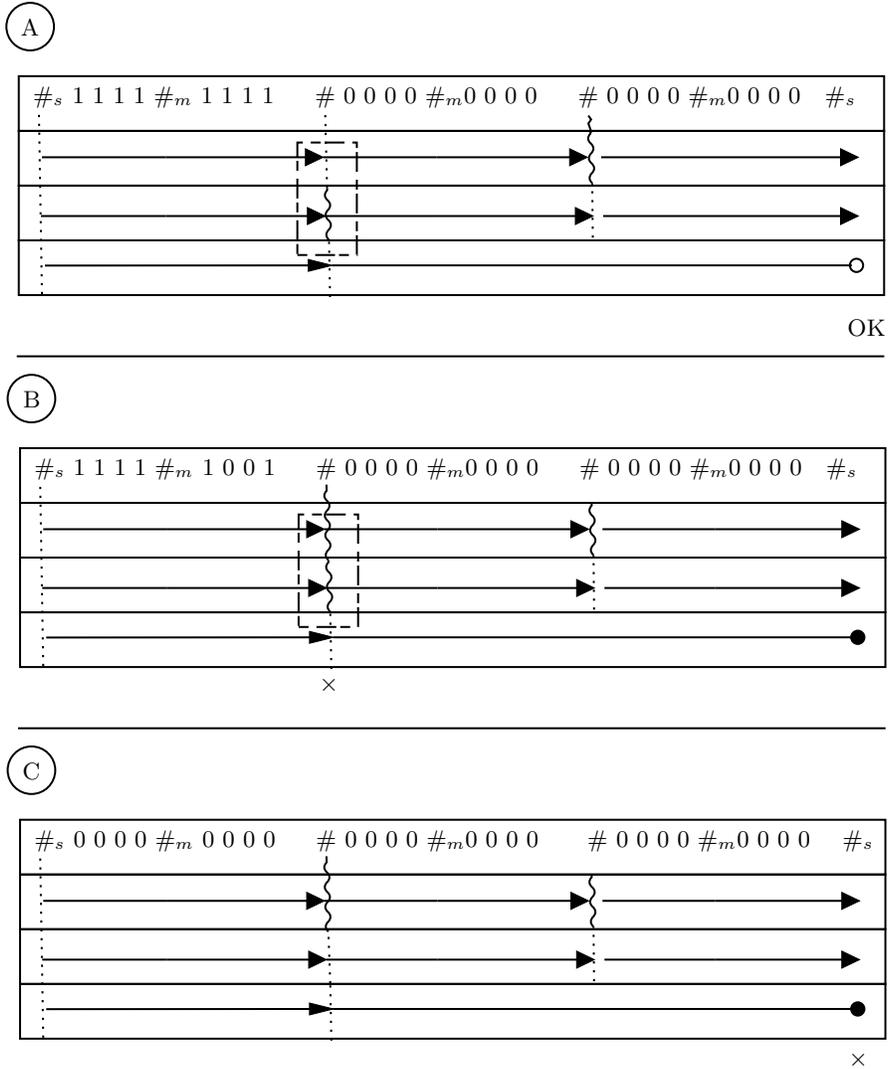
In order to give the detailed construction behind Theorem~\ref{theo:nlhard}, we need some technical definitions. Let $\Gamma$ be a finite set and $\alpha \in \Gamma$. We call a string $y \in \Gamma^{*}$ an $\alpha$-\emph{marker} in $i \in [|y|]$ of length $l \in \N$  if $y|_{[i, i+l]} = \alpha \cdots \alpha.$

Given a directed graph $D=(V,E)$ with $n$ nodes $V=\{v_1,\ldots, v_n\}$ and $m$ edges $E=\{e_1,\ldots, e_m\}$, we consider its (oriented) incidence matrix $M$ defined by 
\[M(i,k) =
  \begin{cases}
    h&\text{ if $e_k=(v_i,\cdot)$ (node $v_i$ is the \emph{head} of edge $e_k$)},\\
    t&\text{ if $e_k=(\cdot,v_i)$ (node $v_i$ is the \emph{tail} of edge $e_k$)},\\
    0&\text{ else.}
  \end{cases}\]
For each $e_k\in E$ (with $1\leq k\leq m$) we define a block representing $e_k$ as a $2 \times (2mn+2m+1)$ matrix $B(e_k)$ such that:

\begin{enumerate}
\item $B(e_k)$ has $2m+1$ special symbols located at specific positions. More precisely $B(e_k)_{i,(n+1)(j-1) + 1} \in  \{\#_s, \#_m, \#\}$ for $i = 1,2,3$ and $j = 1, \hdots, 2m.$
\item  Its first row $B(e_k)_{1,\cdot} \in \{0,1\}^{(2n+2)m+1}$ is a $1$-marker at position $(2n+2)(k-1) + 2$ of length $n$ and a $1-$marker in  $(2n+2)(k-1) + n + 3$ of length $n$; and
\item $B(e_k)_{2,\cdot}\in \{0,1\}^{(2n+2)m+1}$  is a periodic repetition of the row of $M$ corresponding to $e_k$ followed by the same row in reverse order. More precisely:
  \begin{align*}
    B(e_k)_{2,[(2n+2)(i-1)+2, (2n+2)(i-1)+n+1]} &= M_{\cdot,k} \\
    B(e_k)_{2,[(2n+2)(i-1)+n+3, (2n+2)(i-1)+2n+2]} &=  \sigma(M_{\cdot,e_k})
  \end{align*}
for ${1\leq i\leq m}$, and where $\sigma$ is mirror permutation on words, i.e. such that $\sigma(u_1\cdots u_p) = u_p\cdots u_1$.
\end{enumerate}
For an example of a block for some graph $D$ see Figure \ref{fig:blocktext}.

First, observe that  \textsf{REGSPEC} is in \textbf{NL} by Theorem \ref{teo:NL}. Now for the \textbf{NL}-hardness, let us take the problem \textsf{STCON} consisting in given a digraph $D = (N,A)$ and two nodes $s,t \in V$ deciding whether there exists a path connecting $s$ and $t$. Let $(D=(N,A), s, t)$ be an instance of \textsf{STCON}. Observe that any path between $s$ and $t$ can be seen as a sequence of edges $e_{k_1} , \hdots, e_{k_\ell}$ such that $e_{k_1} = (s,v)$, $e_{k_i} = (u',v') \implies e_{k_{i+1}} = (v',w')$ for some $u',v',w' \in N$, $i \in \{2, \hdots \ell -1\}$ and $e_{k_\ell} = (w,t)$ for some $v,w \in N$. Besides, since each edge can be represented by a block then, an $(s,t)$-path $P$ can be represented as a sequences of blocks $B(e_{k_1}), \hdots, B(e_{k_\ell})$. Now, the main idea of the proof is to construct a non-deterministic automata network $\mathcal{A}_{D} = (G_D, \mathcal{F}_D)$ defined over a two dimensional grid of size $k \times d$ where $d = n^{\cO(1)}$ and $k = \cO(1)$. Of course, since $k$ is constant, then $\mathcal{A}_{D}$ has bounded pathwidth. This automata network, will non-deterministically guess a sequences of blocks representing edges in $A$ (selection phase). Needless to say that, the first part of the proof will be showing that $\mathcal{A}_{D} = (G_D, \mathcal{F}_D)$ is capable of deterministically verify if an initial condition corresponds to a sequence of valid edges, i.e.  if it corresponds to a sequence of blocks and they actually represent edges in $A$ (verification phase). In order to perform these task,  we use a construction based on using signals that will collide at specific locations as a way to verify the distance between two given cells, and which use the saved information (as a constant layer) about the incidence matrix of $D$. Generally speaking, once $\mathcal{A}_D$ has verified that the sequence of blocks is valid, it will compare two subsequent blocks $B_{i}$ and $B_{i+1}$ in order to verify that if $B_{i}$ represents edge $(u,v)$ then $B_{i+1}$ represents edge $(v,w)$ for some $u,v,w \in N$. If in any part of its dynamics $\mathcal{A}_D$ detects some error, it will spread an error state that will lead the system to an attractor corresponding to a uniform configuration in which any cell will be in this particular error state. However, if the process runs flawless, then the system will reach an attractor in which all the cells are in a particular success state. We will code, by using a specification $\mathcal{E}_D$ (given in the input of \textsf{SPEC}), a specific requirement for the initial configuration (more precisely, we will ask the initial configuration to have the incidence matrix of $D$, markers and information about the nodes $(s,t)$) in order to allow $\mathcal{A}_D$ to have enough information to start the selection and verification process. In addition, we will code in this specification only the orbits that will reach this specific success state. By doing this, we will show that $\mathcal{A}_D \models \mathcal{E}_D$ if and only if there is a path between $s,t$ in $D$. Thus, the reduction will consist on constructing $(\mathcal{A}_D, \mathcal{E}_D)$ from $(D,s,t)$  in $\DLOG.$ 

Now, we give details on the construction of $\mathcal{E}_D$:
\begin{enumerate}
\item $Q_D = Q_{P}\cup Q_{\text{core}} \cup Q_{\text{signal}}$ where $Q_{\text{core}} = \{\text{Success}, \text{Accept}, \text{Error}, 1, 0, \#_{s},\#_{m},\#\}$, $Q_P = \{P_1, \hdots, P_4\}$ are the states which indicate the different phases that are specified in the paragraph bellow and $Q_{\text{signal}}$ are the states used in order to propagate signals (for example, the ones that we have used them on the previous lemmas).
\item Since it is sufficient to code path without edge repetition, $d$ will be of size at most $m\times b$ where $b = (2n+2)m+1$ is the size of a block. Observe that this size can be fixed since we can always assume that there is a loop in the terminal node $t$ so we can consider that all the paths are coded by $m$ blocks with possible padding of blocks $B((t,t))$. Thus, $d = n^{\cO(1)}.$
\item $k$ will be the number of rows of the grid. We will essentially use one row for the incidence matrix (incidence row), two rows for the blocks (selection row) and a constant number of rows that we will call working rows in which the signals will move and collide (working and verification rows).
\item $\mathcal{E}_D$ will code only initial conditions in which one row of the grid (incidence row) will have $\cO(m)$ copies of the incidence matrix in the same format than the second row in blocks i.e. there are markers at specific positions and we code the different columns in the zones defined by the markers (Figure \ref{fig:veradj}).
\item $\mathcal{E}_D$ will code orbits in which the selection row of size $d$ will have marked in the first block a symbol indicating "head" in the position associated to node $s$ (see Figure \ref{fig:blocktext}).
\item $\mathcal{E}_D$ will code orbits in which the selection row of size $d$  will have marked in the last block a symbol indicating "tail" in the position associated to node $t$ (see Figure \ref{fig:blocktext}).
\item $\mathcal{E}_D$ will code orbits which will reach a uniform success state.
\end{enumerate}

Now we will describe the dynamics of the automata network $\mathcal{A}_D$.  
\paragraph*{Initialization} In $t=0$, since by construction of the specification $\mathcal{E}_D$, we can consider only orbits in which the incidence row, all the special symbols and the position of the source and terminal node (as an $h$ and $t$ symbol fixed in its correspondent positions) are well coded and fixed in the initial condition.  
\paragraph*{Selection phase}
First, the local rules will non-deterministically guess the states of the rest of the cells in the selection row. This process is performed cell by cell, by sending a traveling signal in one of the working rows. This signal starts on a starting symbol $\#_{\text{s}}$ and finishes in a terminal symbol $\#_{\text{s}}$. After doing that, the signal comes back from the terminal symbol to the starting symbol and writes a change of phase state in all the cells on the working row. 
\paragraph*{Verification phase}
After that, verification phase starts. The process has two main subphases:

 \textbf{A local phase:} First, each block is internally verified. More precisely, the local rules will verify that $B(e)$ has the correct formatting on its two rows and that it correspond to an actual edge in $D$, i.e. $e \in A$:
 
\begin{enumerate}
\item \emph{Verification of the first row.} For each part of size $2n$  defined by two different special symbols (i.e. the space bounded by pairs $(\#_s,\#), (\#, \#_s)$ or $(\#,\#)$)), three different signals will start from one symbol to the one in its left (see Figure \ref{fig:secondrowver}). The first signal will change of state if and only if it reads a cell in state $0$. If it remains in initial state it will be interpreted as success otherwise, if it has changed, then it will be interpreted as error. The second one will do the same thing but for cells in state $0$. Finally, the third signal will start from a cell marked with $\#_s$  and will go through the row until another $\#_s$ symbol is reached. This signal will verify that there is exactly one block which has marked success for the first signal and error for the second one. Otherwise, it will change to an error state that will be spread to all the cells (see Figure \ref{fig:secondrowver} for examples).

\item \emph{Verification of the second row.} In order to verify that the coding of the second row is correct, we need to check that each row has exactly two symbols: $h$ and $t$ and that the configuration is symmetric related to the cells marked with symbols $\#_m$. In order to do that, from each symbol $h$ and $t$ a signal is sent through two working rows (one signal to the right and one to the left, see Figure \ref{fig:verdym} for details.) Then, the local rules in the cells holding the state $\#_m$ will change to the success if exactly two signals arrive at the same time. More precisely, this last procedure is implemented by sending a two state signal, one marking the starting part of the signal and one marking the rest. If the latter condition does not hold, the cells marked by $\#_m$  will spread an error state (see Figure \ref{fig:verdym}). Observe that this procedure works since: i) if two cells are holding the same state and they are at the same distance of the cell marked by $\#_m$ then, the two equal signals will arrive at the same time to the cell holding the sate $\#_m$; and ii) since the coding considers a constant amount of special symbols (more precisely $h$ and $t$) then, the local rule is freezing.

\item \emph{Verification of the edge that is coded in the block.} At this point, if no error state has been produced by the dynamics, it means that the coding of each block is coherent, but we are not sure that it actually represents an edge $e \in A$. In fact, we have coded in the first selection row some number $i$ referencing a column of the adjacency matrix of $D$ but, we need to check whether the second selection row contains the same information than the $i$-column of the adjacency matrix. This last part is performed in the following way:  a signal will be transmitted over a working row in order to identify the information in the two selection rows of the block. Since each block has a marker in its first row, the signal can hold a state while it is in the same position than the cells in state one in the marker. Thus, this state will indicate the local rule to perform a comparison between the second row of the block and the correspondent part of the incidence row. For more details see Figure \ref{fig:veradj}. While verifications are being run, the local rule will write an acceptance state or an error state in some working row. Finally, a third signal will verify that all the cells in the latter working row are in the acceptance state and will spread the error state if not. Finally,  if no error state has been spread, the local rule updates the state of the cells in the working row holding the change of phase state.

\begin{figure}
\centering

\tikzset{every picture/.style={line width=0.75pt}} 

\begin{tikzpicture}[x=0.75pt,y=0.75pt,yscale=-1,xscale=1]

\draw   (86,48) -- (521,48) -- (521,69) -- (86,69) -- cycle ;
\draw   (85.5,90) -- (521,90) -- (521,117-2) -- (85.5,117-2) -- cycle ;
\draw   (86,21) -- (521.5,21) -- (521.5,48) -- (86,48) -- cycle ;
\draw   (86,69) -- (521,69) -- (521,90) -- (86,90) -- cycle ;
\draw  [dash pattern={on 0.84pt off 2.51pt}]  (93+5,88) -- (93+5,116) ;
\draw    (97+3,103.5) -- (231-20,103.5) ;
\draw [shift={(234-15,103.5)}, rotate = 180] [fill={rgb, 255:red, 0; green, 0; blue, 0 }  ][line width=0.08]  [draw opacity=0] (8.93,-4.29) -- (0,0) -- (8.93,4.29) -- cycle    ;
\draw  [dash pattern={on 0.84pt off 2.51pt}]  (234-15,90) -- (234-15,117) ;
\draw   (105+6,19) -- (120+8,19) -- (120+8,92) -- (105+6,92) -- cycle ;
\draw   (213,19) -- (228,19) -- (228,92) -- (213,92) -- cycle ;

\draw (88,24.4) node [anchor=north west][inner sep=0.75pt]    {$ \begin{array}{l}
\#_{s} \ h\ t\ 0\ 0\ \#_{m} \ 0\ 0\ t\ h\  \#\ \ 0\ h\ t\ 0\ \#_{m} \ 0\ t\ h\ 0\ \#\ 0\ 0\ h\ t\ \#_{m} \ t\ h\ 0\ 0\ \#\\
\end{array}$};
\draw (88,24.4+25) node [anchor=north west][inner sep=0.75pt]    {$ \begin{array}{l}
\#_{s} \ 1\ 1\ 1\ 1\ \#_{m} \ 1\ 1\ 1\ 1\  \#\ 0\ 0\ 0\ 0\ \#_{m} \ 0\ 0\ 0\ 0\ \#\ 0\ 0\ 0\ 0 \ \  \#_{m} \  0\ 0\ 0\ 0\ \#_{s}
\end{array}$};
\draw (88,24.4+45) node [anchor=north west][inner sep=0.75pt]    {$ \begin{array}{l}
\#_{s} \ h\ t\ 0\ 0\ \#_{m} \ 0\ 0\ t\ h\ \#\ h\ t\ 0\ 0\ \#_{m} \ 0\ 0\ t\ h\ \#\ h\  t\ 0\ 0\ \#_{m} \ 0\ 0\ t\ h\ \#_{s}
\end{array}$};

\end{tikzpicture}
\caption{Example of adjacency verification. In the first row, the incidence matrix of $D$ is coded. In this case a signal verifies that the edge $(1,2)$ is in the graph $D$ (see Figure \ref{fig:blocktext})}
\label{fig:veradj}
\end{figure}
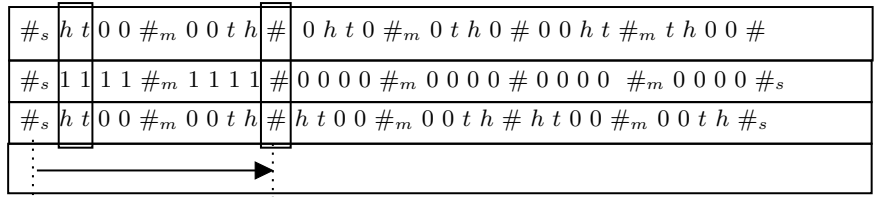

\end{enumerate}

 \textbf{A pair-wise coherent phase:}  
 similarly to the verification of the second row in the selection row on the previous phase, this phase starts by sending multiple signals that are sent from the cells in the selection row with states given by the symbols marking the tails and the heads of the coded edge in the second row of the selection row on each block. These signals are sent through two different working tapes. Each of these signals will carry a special state indicating if its origin was a head or a tail. The local rule in the cells with a special symbol ($\#_s$) will verify whether a head signal has collided with a tail signal (see Figure \ref{fig:verblocks}). If exactly one of this collision take place, the local rule will write an accept state in one of the working rows. Finally, in other working  row, a signal starting from the starting symbol will verify that at the position of the beginning (ending) of a block an accept state is written in the previous working row. The local rule will update  the cells in that working row to an error state that will spread if at least one the verifications is not correct. Otherwise, it will update the cells to the success state.

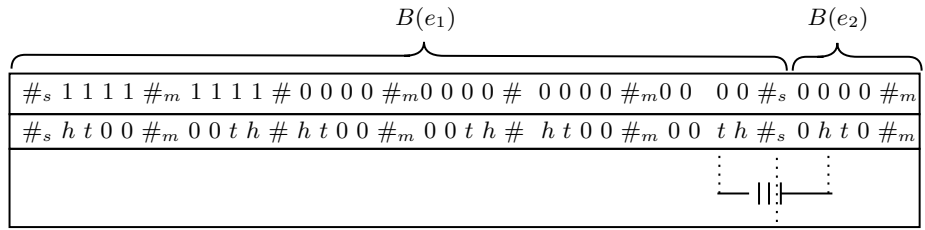
\begin{figure}

\tikzset{every picture/.style={line width=0.75pt}} 

\begin{tikzpicture}[x=0.75pt,y=0.75pt,yscale=-1,xscale=1]

\draw   (66,58.47+5) -- (578-53,58.47+5) -- (578-53,85.96-2) -- (66,85.96-2) -- cycle ;
\draw   (66,85.96-2) -- (578-53,85.96-2) -- (578-53,113.46-12) -- (66,113.46-12) -- cycle ;
\draw   (66,113.46-12) -- (578-53,113.46-12) -- (578-53,140.95) -- (66,140.95) -- cycle ;
\draw  [dash pattern={on 0.84pt off 2.51pt}]  (488-35,106) -- (488-35,140) ;
\draw  [dash pattern={on 0.84pt off 2.51pt}]  (462-38,105) -- (462-38,124) ;
\draw  [dash pattern={on 0.84pt off 2.51pt}]  (519-40,105) -- (519-40,124) ;
\draw    (461-38,124) -- (477-38,124) ;
\draw [shift={(477-33,124)}, rotate = 180] [color={rgb, 255:red, 0; green, 0; blue, 0 }  ][line width=0.75]    (0,5.59) -- (0,-5.59)(-5.03,5.59) -- (-5.03,-5.59)   ;
\draw    (493-38,124) -- (519-38,124) ;
\draw [shift={(493-38,124)}, rotate = 180] [color={rgb, 255:red, 0; green, 0; blue, 0 }  ][line width=0.75]    (0,5.59) -- (0,-5.59)   ;
\draw   (483-25,52+10) .. controls (483.01-25,47.33+10) and (480.69-25,44.99+10) .. (476.02-25,44.98+10) -- (285.02,44.52+10) .. controls (278.35,44.5+10) and (275.02,42.16+10) .. (275.03,37.49+10) .. controls (275.02,42.16+10) and (271.69,44.48+10) .. (265.02,44.47+10)(268.02,44.47+10) -- (74.02,44+10) .. controls (69.35,43.99+10) and (67.01,46.31+10) .. (67,50.98+10) ;
\draw   (571-45,52+10) .. controls (571-45,47.33+10) and (568.67-45,45+10) .. (564-45,45+10) -- (538.5-45,45+10) .. controls (531.83-45,45+10) and (528.5-45,42.67+10) .. (528.5-45,38+10) .. controls (528.5-45,42.67+10) and (525.17-45,45+10) .. (518.5-45,45+10)(521.5-45,45+10) -- (493-25,45+10) .. controls (488.33-25,45+10) and (486-25,47.33+10) .. (486-25,52+10) ;

\draw (69,64.4) node [anchor=north west][inner sep=0.75pt]    {$ \begin{array}{l}
\#_{s} \ 1\ 1\ 1\ 1\ \#_{m} \ 1\ 1\ 1\ 1\ \#\ 0\ 0\ 0\ 0\ \#_{m} 0\ 0\ 0\ 0\ \#\ \ 0\ 0\ 0\ 0\ \#_{m} 0\ 0\ \ \ 0\ 0 \ \#_{s}\\
\end{array}$};
\draw (69,64.4+20) node [anchor=north west][inner sep=0.75pt]    {$ \begin{array}{l}
\#_{s} \ h\ t\ 0\ 0\ \#_{m} \ 0\ 0\ t\ h\ \#\ h\ t\ 0\ 0\  \#_{m} \ 0\ 0\ t\ h\ \#\ \ h\  t\ 0\ 0\ \#_{m} \ 0\ 0\ \  t\ h \ \#_{s}
\end{array}$};
\draw (501-45,64.4) node [anchor=north west][inner sep=0.75pt]    {$ \begin{array}{l}
\ 0\ 0\ 0\ 0\ \#_{m} \ \\
\end{array}$};
\draw (501-45,64.4+20) node [anchor=north west][inner sep=0.75pt]    {$ \begin{array}{l}
\ 0\ h\ t\ 0\ \#_{m}
\end{array}$};
\draw (259,13.4+15) node [anchor=north west][inner sep=0.75pt]    {$B( e_{1})$};
\draw (512-45,13.4+15) node [anchor=north west][inner sep=0.75pt]    {$B( e_{2})$};
\end{tikzpicture}
\caption{Example of a verification dynamics which compares blocks and checks if the corresponding edges are both incident to the same node. If exactly one tail signal  collide with a head it means that the previous property is verified and an error state is spread otherwise.}
\label{fig:verblocks}
\end{figure}
 
We can now show that Theorem~\ref{theo:nlhard} holds.

\begin{proof}[Proof of Theorem~\ref{theo:nlhard}]
First, observe that the construction of $(\mathcal{A}_D, \mathcal{E}_D)$ can be done in $\DLOG$ since local rules does not depend on the structure of $D$ and thus, we only need to store partial information related to the structure of the incidence matrix of $D$ in order to define the specification. Then, we have that if there is a path between $s$ and $t$ on $D$, by construction, there must be at least one orbit of $\mathcal{A}_D$ which satisfies $\mathcal{E}_D$. Conversely, if $\mathcal{A}_D \models \mathcal{E}_D$ then, there exist at least one initial condition which codes a sequence of edges in $D$ which leads the system to a uniform success fixed point. By construction, this attractor is only reachable (starting from the set of valid initial conditions) after all the previous phases are successfully performed by the dynamics. Then, we deduce that $\textsf{STCON} \leq^{\DLOG}_{m} \textsf{SPEC}$ and thus,  \textsf{SPEC} is $\textbf{NL}$-hard. \qed
\end{proof}



\section{Hardness of $FO^+$ model checking}

${FO(=,\to)}$ denotes the first order logic over configurations using equality and a predicate ${x\to y}$ meaning that $y$ can be reached from $x$ in one step. It is well-known that this logic can be efficiently dealt with using finite automata theory when configurations are one-dimensional. For instance the model-checking of this logic is decidable on one-dimensional CAs \cite{Finkel11,Sutner09}. 
In this subsection, we study first-order properties of the dynamics enriched with a new predicate $x \to^+ y$ expressing that configuration $y$ can be reached from configuration $x$ in some unknown number of steps. We denote this logic $FO^+=FO(=,\to,\to^+)$. Adding the predicate $\to^+$ allows to express properties like nilpotency:
\[\exists x, \forall z, (z\to^+ z)\implies z=x,\]
which is equivalent in the deterministic case to ${\exists x,\forall y, y\to^+x}$.
The model checking of $FO^+$ is therefore undecidable for general 1D CA \cite{kari92} and \textbf{PSPACE-complete} for AN of bounded pathwidth \cite{stacs.2021.32}. However, nilpotency is a decidable property for 1D freezing CA \cite{dmtcs:9004} and \textbf{co-NL} for bounded pathwidth freezing AN (Corollary~\ref{cor:nilconl}). It is therefore interesting to figure out what is the complexity of the model checking of $FO^+$ for freezing AN of bounded pathwidth.

The goal of this section is to show that despite considering only ``one-dimensional'' networks and having the  freezing constraint, we can encode bi-dimensional domino problems in $FO^+$ and thus get a \textbf{NP-hard} lower bound. The precise NP-hard problem we consider in this subsection to reduce from is the following.

\begin{lemma}[HV-domino CSP]\label{lem:hvdomino}
  Let $Q$ be a large enough alphabet. The following problem is NP-complete:
  \begin{itemize}
  \item \textbf{input:} for each ${1\leq i,j\leq n}$, two lists of constraints: ${H_{i,j}\subseteq Q^2}$ and ${V_{i,j}\subseteq Q^2}$.
  \item \textbf{question:} does there exist a configuration ${a\in Q^{\{1,\ldots,n\}\times\{1,\ldots,n\}}}$ such that for all ${1\leq i,j\leq n}$ the local constraints are satisfied, \textit{i.e.} 
    \[(a_{i,j},a_{i+1,j})\in H_{i,j}\text{ (if $i<n$) and }(a_{i,j},a_{i,j+1})\in V_{i,j}\text{ (if $j<n$)}.\]
  \end{itemize}
\end{lemma}
\begin{proof}
  There exists a Turing machine working in polynomial time (and space) that on input ${(\Psi,v)}$ where $\Psi$ is a SAT formula and $v$ a candidate valuation checks whether $v$ satisfies $\Psi$.
  Then for any given SAT formula $\Psi$, one can produce in LOGSPACE a set of HV-domino constraints that accepts only bi-dimensional configurations ${(a_{i,j})}$ which represent a valid space-time diagram of the above machine which are correctly initialized and with $\Psi$ enforced as the first component of the input. The encoding of space-time diagram of Turing machine inside domino constraints is well-known and usually presented through a fixed set of so-called Wang tiles (see for example~\cite{Ro71}), which are just a uniform set of horizontal constraints ${H\subseteq Q^2}$ and vertical constraints ${V\subseteq Q^2}$.
  Note that since the HV-domino constraints considered here are non-uniform, we can hard-code the initial state of the machine in the lower-left corner of the configuration, the encoding of $\Psi$ in the initial row, and the accepting state of the machine in the top row.
  The reduction from SAT to the HV-domino CSP follows.\qed
\end{proof}

We can now show a lower bound on the model checking of a single formula of $FO^+$. Let $P_k(x)$ denote the formula of $FO^+$ expressing that configuration $x$ has at least $k$ preimages, formally:
  $$P_k(x) \equiv \exists x_1\neq x_2\neq\cdots \neq x_k, \bigwedge_{1\leq i\leq k}x_i\to x.$$

We will consider the following formula $\phi$:
\begin{align*}
  \phi&\equiv \exists x: x\to x\\ &\wedge \bigl(\forall y,\forall y^1,\forall z, (\neg P_1(y)\wedge\neg P_2(y^1)\wedge y\to y^1\wedge y^1\to^+z\wedge z\to x\wedge z\neq x)\Rightarrow \neg P_2(z)\bigr).
\end{align*}
It expresses that there exists a fixed point $x$ such that considering any orbit starting from a configuration $y$ without preimage, which is the unique preimage of its successor $y^1$, and leading to $x$, then  the configuration $z$ occurring in the orbit just before reaching $x$ has only $1$ preimage.

The main result of this section is that the $FO^+$ model checking problem is already hard for formula $\phi$.
The proof uses the HV-domino CSP of Lemma~\ref{lem:hvdomino}.
For each HV-domino CSP problem, we build a deterministic one-dimensional freezing automata network that essentially checks that a configuration $(a_{i,j})$ satisfies the HV-domino constraints.
By one-dimensional we mean a graph which is a line with self-loops on each node.
In this automata network, configurations $(a_{i,j})$ are layed out as one-dimensional configurations so that $a_{i,j}$ and $a_{i+1,j}$ are neighbors in the graph, and therefore H-constraints can be checked locally.
However, $a_{i,j}$ and $a_{i,j+1}$ are far away in the graph, so V-constraints require the dynamics of the automata network to be checked.
The key idea is to use formula $\phi$ above to characterize the part of the dynamics of the automata network that checks all V-constraints for a given candidate configuration $(a_{i,j})$: intuitively, quantifying over orbits starting from a configuration $y$ without preimage and being the unique preimage of its successor ensures that the orbit contains some well-initialized computation, and predicate $\neg P_2$ on the configuration before reaching the fixed point codes the fact that the output of the computation is correct.
The fixed point configuration $x$ in formula $\phi$ represents a candidate configuration ${(a_{i,j})}$ (cleaned from any trace of computation) and, by construction of the automata network, the second part of the formula expresses that for any well-initialized test of a V-constraint the output of the test is a success.
Formula $\phi$ uses predicate $\neg P_2$ to characterize some specific configurations: the key corresponding trick in the construction below is to make Cartesian products of some alphabet with ${\{0,1\}}$ and ensure that the action of the automata network almost always reset to 1 the value of such a ${\{0,1\}}$-component in at least one node.
This ensures that the configuration obtained after one step has at least two preimages.
The situations where it is not the case are exceptional and well-controlled: this helps to identify possible candidates for configurations $y$, $y^1$ and $z$ in formula $\phi$.

\begin{theorem}\label{theo:fop-hardness}
  Checking whether a given deterministic freezing automata network $(G,\mathcal{F})$ verifies $\phi$ is NP-hard, even when restricted to bounded alphabet, and degree 3 and pathwidth 1.
\end{theorem}

\begin{proof}
  We proceed by reduction from the HV-domino CSP: given $n$ and constraints $(H_{i,j})$ and $(V_{i,j})$, we build a deterministic automata network ${(G_N,\mathcal{F})}$ with ${N=n^2}$ which verifies $\phi$ if and only if the CSP has a solution.
  $G_N$ is the graph with nodes ${V=\{1,\ldots,N\}}$ and edges ${(i,i)}$ for all $i\in V$ and ${(i,i+1)}$ for all $i<N$ and ${(i,i-1)}$ for all $i>0$.
  $G_N$ has pathwidth 1 and degree 3.
  The automata network $\mathcal{F}$ has four components plus a global error state and uses alphabet ${Q'=Q\times\{0,1\}\times Q_h\times Q_t\cup\{\bot\}}$ (where $Q$ is the alphabet of the HV-domino CSP).
  The freezing order on $Q\times\{0,1\}\times Q_h\times Q_t$ is simply the product of orders on each component, and this order is extended to $Q'$ by taking $\bot$ as a maximal element.
  The overall behavior is as follows (see Figure~\ref{fig:components} and Figure~\ref{fig:euclint}).

  \newcommand\stnm[4]{{\draw[fill=#4] (#2,#3)
    -- ++(1,0) -- ++(0,1)--++(-1,0)--cycle;\draw (#2,#3)+(.5,.5) node {\tiny #1};}}

\newcommand\stateZ[2]{\stnm{$A$}{#1}{#2}{yellow!50!white}}
\newcommand\stateO[2]{\stnm{$\rightarrow$}{#1}{#2}{white}}
\newcommand\stateZO[2]{\stnm{$\rightarrow_1$}{#1}{#2}{red!50!white}}
\newcommand\stateOO[2]{\stnm{$\leftarrow_1$}{#1}{#2}{red!50!white}}
\newcommand\stateZZO[2]{\stnm{$R_1$}{#1}{#2}{red!70!white}}
\newcommand\stateOZO[2]{\stnm{$C$}{#1}{#2}{blue!50!white}}
\newcommand\stateZOO[2]{\stnm{$B$}{#1}{#2}{green!50!white}}
\newcommand\stateOOO[2]{\draw[fill=white!87!black] (#1,#2) rectangle +(1,1);}

  \newcommand\cadre[4]{{\draw[very thick,dotted] (#2,#3)
    -- ++(1,0) -- ++(0,1)--++(-1,0)--cycle;}}
\newcommand\cadreO[2]{\cadre{}{#1}{#2}{}}
\newcommand\cadreZO[2]{\cadre{}{#1}{#2}{}}
\newcommand\cadreOO[2]{\cadre{}{#1}{#2}{}}

\begin{figure}
  \begin{minipage}{.45\linewidth}
    \begin{center}
      \begin{tikzpicture}[scale=.5]
        \draw[->] (-1,0)-- node[midway,sloped,above] {time} (-1,5);
        \stateO{0}{0}\stateZOO{0}{1}\stateZOO{0}{2}\stateZOO{0}{3}\stateZOO{0}{4}\stateZOO{0}{5}\stateZOO{0}{6}\stateZOO{0}{7}\stateZOO{0}{8}\stateZOO{0}{9}\stateZOO{0}{10}\stateZOO{0}{11}\stateZOO{0}{12}\stateZOO{0}{13}\stateZOO{0}{14}\stateZOO{0}{15}\stateZOO{0}{16}\stateOO{0}{17}\stateZZO{0}{18}\stateOZO{0}{19}\stateZ{1}{0}\stateO{1}{1}\stateZOO{1}{2}\stateZOO{1}{3}\stateZOO{1}{4}\stateZOO{1}{5}\stateZOO{1}{6}\stateZOO{1}{7}\stateZOO{1}{8}\stateZOO{1}{9}\stateZOO{1}{10}\stateZOO{1}{11}\stateZOO{1}{12}\stateZOO{1}{13}\stateZOO{1}{14}\stateZOO{1}{15}\stateOO{1}{16}\stateOZO{1}{17}\stateOZO{1}{18}\stateOZO{1}{19}\stateZ{2}{0}\stateZ{2}{1}\stateO{2}{2}\stateZOO{2}{3}\stateZOO{2}{4}\stateZOO{2}{5}\stateZOO{2}{6}\stateZOO{2}{7}\stateZOO{2}{8}\stateZOO{2}{9}\stateZOO{2}{10}\stateZOO{2}{11}\stateZOO{2}{12}\stateZOO{2}{13}\stateZOO{2}{14}\stateOO{2}{15}\stateOZO{2}{16}\stateOZO{2}{17}\stateOZO{2}{18}\stateOZO{2}{19}\stateZ{3}{0}\stateZ{3}{1}\stateZ{3}{2}\stateO{3}{3}\stateZOO{3}{4}\stateZOO{3}{5}\stateZOO{3}{6}\stateZOO{3}{7}\stateZOO{3}{8}\stateZOO{3}{9}\stateZOO{3}{10}\stateZOO{3}{11}\stateZOO{3}{12}\stateZOO{3}{13}\stateOO{3}{14}\stateOZO{3}{15}\stateOZO{3}{16}\stateOZO{3}{17}\stateOZO{3}{18}\stateOZO{3}{19}\stateZ{4}{0}\stateZ{4}{1}\stateZ{4}{2}\stateZ{4}{3}\stateZO{4}{4}\stateZOO{4}{5}\stateZOO{4}{6}\stateZOO{4}{7}\stateZOO{4}{8}\stateZOO{4}{9}\stateZOO{4}{10}\stateZOO{4}{11}\stateZOO{4}{12}\stateOO{4}{13}\stateOZO{4}{14}\stateOZO{4}{15}\stateOZO{4}{16}\stateOZO{4}{17}\stateOZO{4}{18}\stateOZO{4}{19}\stateZ{5}{0}\stateZ{5}{1}\stateZ{5}{2}\stateZ{5}{3}\stateZ{5}{4}\stateZO{5}{5}\stateZOO{5}{6}\stateZOO{5}{7}\stateZOO{5}{8}\stateZOO{5}{9}\stateZOO{5}{10}\stateZOO{5}{11}\stateOO{5}{12}\stateOZO{5}{13}\stateOZO{5}{14}\stateOZO{5}{15}\stateOZO{5}{16}\stateOZO{5}{17}\stateOZO{5}{18}\stateOZO{5}{19}\stateZ{6}{0}\stateZ{6}{1}\stateZ{6}{2}\stateZ{6}{3}\stateZ{6}{4}\stateZ{6}{5}\stateZO{6}{6}\stateZOO{6}{7}\stateZOO{6}{8}\stateZOO{6}{9}\stateZOO{6}{10}\stateOO{6}{11}\stateOZO{6}{12}\stateOZO{6}{13}\stateOZO{6}{14}\stateOZO{6}{15}\stateOZO{6}{16}\stateOZO{6}{17}\stateOZO{6}{18}\stateOZO{6}{19}\stateZ{7}{0}\stateZ{7}{1}\stateZ{7}{2}\stateZ{7}{3}\stateZ{7}{4}\stateZ{7}{5}\stateZ{7}{6}\stateZO{7}{7}\stateZOO{7}{8}\stateZOO{7}{9}\stateOO{7}{10}\stateOZO{7}{11}\stateOZO{7}{12}\stateOZO{7}{13}\stateOZO{7}{14}\stateOZO{7}{15}\stateOZO{7}{16}\stateOZO{7}{17}\stateOZO{7}{18}\stateOZO{7}{19}\stateZ{8}{0}\stateZ{8}{1}\stateZ{8}{2}\stateZ{8}{3}\stateZ{8}{4}\stateZ{8}{5}\stateZ{8}{6}\stateZ{8}{7}\stateZO{8}{8}\stateOO{8}{9}\stateOZO{8}{10}\stateOZO{8}{11}\stateOZO{8}{12}\stateOZO{8}{13}\stateOZO{8}{14}\stateOZO{8}{15}\stateOZO{8}{16}\stateOZO{8}{17}\stateOZO{8}{18}\stateOZO{8}{19}
        \draw[thick] (3,0) -- (3,20);
        \draw[thick] (6,0) -- (6,20);
      \end{tikzpicture}\vskip .3cm
      $Q_h$ component
    \end{center}
  \end{minipage}
\renewcommand\stateZ[2]{\stnm{$L$}{#1}{#2}{blue!50!white}}
\renewcommand\stateO[2]{\stnm{$\alpha_L$}{#1}{#2}{blue!70!white}}
\renewcommand\stateZO[2]{\stnm{$\alpha_R$}{#1}{#2}{red!70!white}}
\renewcommand\stateOO[2]{\stnm{$\leftarrow_L$}{#1}{#2}{green!50!white}}
\renewcommand\stateZZO[2]{\stnm{$\rightarrow_L$}{#1}{#2}{green!70!white}}
\renewcommand\stateOZO[2]{\stnm{$\leftarrow_R$}{#1}{#2}{green!70!white}}
\renewcommand\stateZOO[2]{\stnm{$\beta$}{#1}{#2}{black!50!white}}
\renewcommand\stateOOO[2]{\stnm{$\gamma$}{#1}{#2}{black!20!white}}
\newcommand\stateZZZO[2]{\stnm{$R$}{#1}{#2}{red!50!white}}
  \begin{minipage}{.45\linewidth}
    \begin{center}
      \begin{tikzpicture}[scale=.5]
        \draw[->] (-1,0)-- node[midway,sloped,above] {time} (-1,5);
        \stateZ{0}{0}\stateZ{0}{1}\stateZ{0}{2}\stateZ{0}{3}\stateZ{0}{4}\stateZ{0}{5}\stateZ{0}{6}\stateZ{0}{7}\stateZ{0}{8}\stateZ{0}{9}\stateZ{0}{10}\stateZ{0}{11}\stateZ{0}{12}\stateZ{0}{13}\stateZ{0}{14}\stateZ{0}{15}\stateZ{0}{16}\stateOOO{0}{17}\stateOOO{0}{18}\stateOOO{0}{19}\stateZ{1}{0}\stateZ{1}{1}\stateZ{1}{2}\stateZ{1}{3}\stateZ{1}{4}\stateZ{1}{5}\stateZ{1}{6}\stateZ{1}{7}\stateZ{1}{8}\stateZ{1}{9}\stateZ{1}{10}\stateZ{1}{11}\stateZ{1}{12}\stateZ{1}{13}\stateZ{1}{14}\stateZ{1}{15}\stateOOO{1}{16}\stateOOO{1}{17}\stateOOO{1}{18}\stateOOO{1}{19}\stateZ{2}{0}\stateZ{2}{1}\stateZ{2}{2}\stateZ{2}{3}\stateZ{2}{4}\stateZ{2}{5}\stateZ{2}{6}\stateZ{2}{7}\stateZ{2}{8}\stateZ{2}{9}\stateZ{2}{10}\stateZ{2}{11}\stateZ{2}{12}\stateZ{2}{13}\stateZ{2}{14}\stateOOO{2}{15}\stateOOO{2}{16}\stateOOO{2}{17}\stateOOO{2}{18}\stateOOO{2}{19}\stateZ{3}{0}\stateZ{3}{1}\stateZ{3}{2}\stateZ{3}{3}\stateZ{3}{4}\stateOO{3}{5}\stateZZO{3}{6}\stateOOO{3}{7}\stateOOO{3}{8}\stateOOO{3}{9}\stateOOO{3}{10}\stateOOO{3}{11}\stateOOO{3}{12}\stateOOO{3}{13}\stateOOO{3}{14}\stateOOO{3}{15}\stateOOO{3}{16}\stateOOO{3}{17}\stateOOO{3}{18}\stateOOO{3}{19}\stateO{4}{0}\stateO{4}{1}\stateO{4}{2}\stateO{4}{3}\stateOO{4}{4}\stateZOO{4}{5}\stateZOO{4}{6}\stateZZO{4}{7}\stateOOO{4}{8}\stateOOO{4}{9}\stateOOO{4}{10}\stateOOO{4}{11}\stateOOO{4}{12}\stateOOO{4}{13}\stateOOO{4}{14}\stateOOO{4}{15}\stateOOO{4}{16}\stateOOO{4}{17}\stateOOO{4}{18}\stateOOO{4}{19}\stateO{5}{0}\stateO{5}{1}\stateO{5}{2}\stateO{5}{3}\stateO{5}{4}\stateO{5}{5}\stateO{5}{6}\stateO{5}{7}\stateZZO{5}{8}\stateZZO{5}{9}\stateZZO{5}{10}\stateZZO{5}{11}\stateOOO{5}{12}\stateOOO{5}{13}\stateOOO{5}{14}\stateOOO{5}{15}\stateOOO{5}{16}\stateOOO{5}{17}\stateOOO{5}{18}\stateOOO{5}{19}\stateZO{6}{0}\stateZO{6}{1}\stateZO{6}{2}\stateZO{6}{3}\stateZO{6}{4}\stateZO{6}{5}\stateZO{6}{6}\stateZO{6}{7}\stateOZO{6}{8}\stateOZO{6}{9}\stateOZO{6}{10}\stateOOO{6}{11}\stateOOO{6}{12}\stateOOO{6}{13}\stateOOO{6}{14}\stateOOO{6}{15}\stateOOO{6}{16}\stateOOO{6}{17}\stateOOO{6}{18}\stateOOO{6}{19}\stateZO{7}{0}\stateZO{7}{1}\stateZO{7}{2}\stateZO{7}{3}\stateZO{7}{4}\stateZO{7}{5}\stateZO{7}{6}\stateOZO{7}{7}\stateOOO{7}{8}\stateOOO{7}{9}\stateOOO{7}{10}\stateOOO{7}{11}\stateOOO{7}{12}\stateOOO{7}{13}\stateOOO{7}{14}\stateOOO{7}{15}\stateOOO{7}{16}\stateOOO{7}{17}\stateOOO{7}{18}\stateOOO{7}{19}\stateZZZO{8}{0}\stateZZZO{8}{1}\stateZZZO{8}{2}\stateZZZO{8}{3}\stateZZZO{8}{4}\stateZZZO{8}{5}\stateZZZO{8}{6}\stateZZZO{8}{7}\stateZZZO{8}{8}\stateOOO{8}{9}\stateOOO{8}{10}\stateOOO{8}{11}\stateOOO{8}{12}\stateOOO{8}{13}\stateOOO{8}{14}\stateOOO{8}{15}\stateOOO{8}{16}\stateOOO{8}{17}\stateOOO{8}{18}\stateOOO{8}{19}
        \cadreO{0}{0}\cadreOO{0}{17}\cadreO{1}{1}\cadreOO{1}{16}\cadreO{2}{2}\cadreOO{2}{15}\cadreO{3}{3}\cadreOO{3}{14}\cadreZO{4}{4}\cadreOO{4}{13}\cadreZO{5}{5}\cadreOO{5}{12}\cadreZO{6}{6}\cadreOO{6}{11}\cadreZO{7}{7}\cadreOO{7}{10}\cadreZO{8}{8}\cadreOO{8}{9}
        \draw[very thick] (3,0) -- (3,20);
        \draw[very thick] (6,0) -- (6,20);
      \end{tikzpicture}\vskip .3cm
      $Q_t$ component
    \end{center}
  \end{minipage}
  \caption{Example of valid orbit with ${n=3}$ starting from a configuration testing the V-constraint $V_{2,2}$ and resulting in a positive output. The trajectory of the $Q_h$ head is reproduced on the $Q_t$-component to clarify the interactions. The vertical thick lines represent separations between consecutive blocks.\label{fig:components}}
\end{figure}
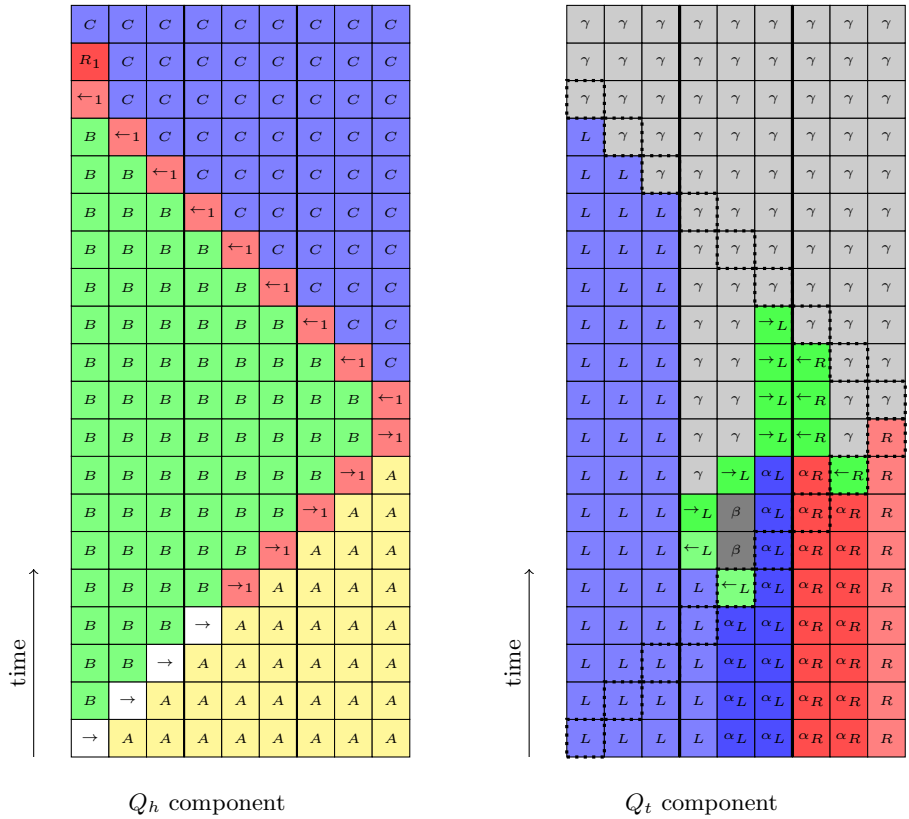
    \begin{figure}\centering
      \begin{tikzpicture}[scale=.45]
        \draw (0,0) -- (9,0) -- (9,18) -- (0,18) -- cycle;
        \draw[thick] (3,0) -- (3,18);
        \draw[thick,color=blue] (5,0) node[below] {\tiny left mark} -- +(0,18);
        \draw[thick] (6,0)  -- (6,18);
        \draw[thick,color=red] (8,0) node[below] {\tiny right mark} -- +(0,18);
        \draw[very thick,dotted,color=gray] (0,0) -- (9,9) -- (0,18);
        \draw[->,>=latex] (12,2) node[right] {\small left/right heads of the $Q_t$-component} -- (4,6);
        \draw[->,>=latex] (12,2) -- (7,9);
        \draw[->,>=latex] (12,14) node[right] {\small head of $Q_h$-component} -- (7,11);
        \draw[very thick,dotted,color=green] (5,5) -- (3,7) -- (6,10);
        \draw[very thick,dotted,color=green] (8,8) -- (6,10);
        \draw[->,>=latex] (12,9) node[right] {\small test success iff heads meet on block boundary} -- (6,10);
        \draw[->,>=latex] (12,18) node[right] {\small block boundaries} -- (6,15);
        \draw[->,>=latex] (12,18) -- (3,17);
      \end{tikzpicture}
      \caption{Euclidean non-discretized representation of the verification process ensuring that the marks in two consecutive blocks have the same offset within the block and thus mark two positions which are vertical neighbors in the matrix ${(a_{i,j})}$, \textit{i.e.} of the form $a_{i,j}$ and $a_{i+1,j}$.\label{fig:euclint}}
    \end{figure}
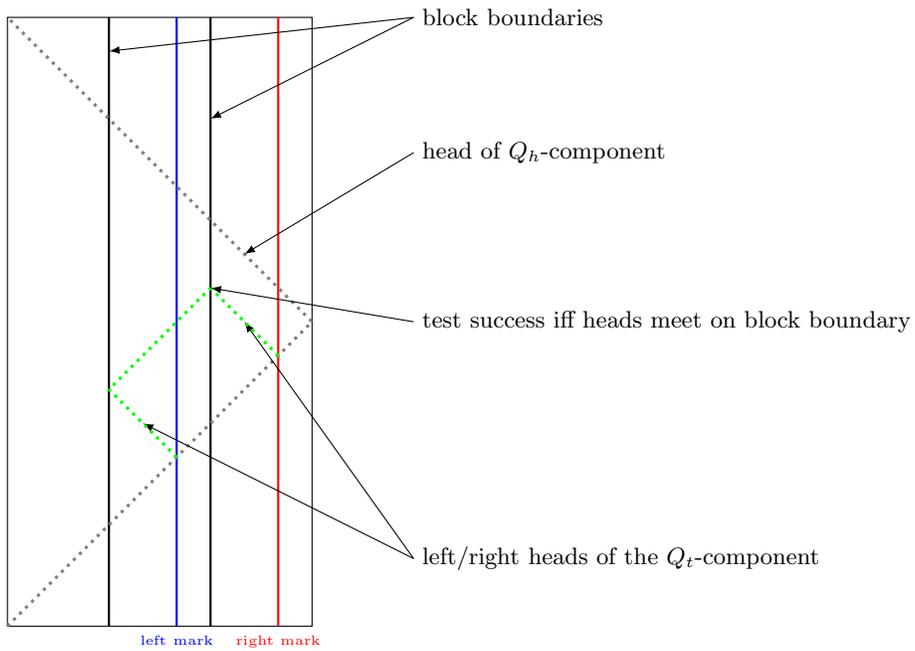

  \begin{itemize}
  \item $\bot$ is an invariable spreading error state: as soon as some node is in state $\bot$, its neighbors change to $\bot$ in one step.
  \item The $Q$-component contains a candidate configuration ${(a_{i,j})}$ written as a one-dimensional word 
    ${a_{1,1}\cdots a_{1,n}a_{2,1}\cdots a_{2,n}\cdots a_{n,1}\cdots a_{n,n}}$.
    The block of nodes ${(j-1)n+1}$ to ${(j-1)n+n}$ will be referred to as \emph{block $j$} and it contains line $j$ of the matrix ${(a_{i,j})}$ in its $Q$-component.
    This component never changes, except when an error state $\bot$ invades the network, or when some H-constraint $H_{i,j}$ is violated at some node in which case a $\bot$ state is generated.
    The freezing order on this component can be chosen arbitrarily.
  \item The $\{0,1\}$-component is called \emph{dummy component} which never changes, has no influence on other components, and is just here to ensure that any configuration leading to $\bot^N$ has enough preimages (see Claim~\ref{claim:errorpreimages} below).
  \item The $Q_h$ component handles a global control \emph{head} whose main behavior is a back-and-forth movement from node $1$ to node $N$ and back to node $1$.
    More precisely, the head do so on a set $\Sigma$ of well-formed configurations and any ill-formed configuration is detected locally and generates an error state $\bot$.
    We set ${Q_h=\{0,1\}\times\{0,1\}\times\{A,B,C,\rightarrow,\rightarrow_0,\rightarrow_1,\leftarrow_1,\leftarrow_0,R_0,R_1\}}$ with freezing order ${0<1}$ on the two $\{0,1\}$ components and \[A<\rightarrow<\rightarrow_0<\rightarrow_1<B<\leftarrow_1<\leftarrow_0<R_0<R_1<C\]
 on the remaining component.
    $\Sigma$ is defined by forbidding a set of pairs of states to occur two adjacent symbols ${c_ic_{i+1}}$ from the third component.
    Moreover, we add the constraint that node $1$ cannot be in state $A$, and that the first $\{0,1\}$ component at this node must be $1$.
    Precisely, configurations authorized in $\Sigma$ are the following (without considering the $\{0,1\}$ components):
    \begin{enumerate}
    \item ${\rightarrow A^{N-1}}$ or ${\rightarrow_0 A^{N-1}}$ or ${\rightarrow_1 A^{N-1}}$,
    \item ${B^i\rightarrow A^{N-i-1}}$ or ${B^i\rightarrow_0 A^{N-i-1}}$ or ${B^i\rightarrow_1 A^{N-i-1}}$,
    \item ${B^{N-1}\rightarrow_0}$ or ${B^{N-1}\rightarrow_1}$,
    \item ${B^{N-1}\leftarrow_0}$ or ${B^{N-1}\leftarrow_1}$,
    \item ${B^i\leftarrow_0 C^{N-i-1}}$ or ${B^i\leftarrow_1 C^{N-i-1}}$,
    \item ${\leftarrow_1 C^{N-1}}$ or ${\leftarrow_0 C^{N-1}}$,
    \item ${R_1 C^{N-1}}$ or ${R_0 C^{N-1}}$,
    \item ${C^N}$.
    \end{enumerate}
    The head is the unique arrow occurring in each configuration and its dynamics is as follows.
    It moves to the right in a background of As and letting symbols B behind (configuration types 1, 2 and 3).
    At each move to the right, the first ${\{0,1\}}$ component of the node left by the head is reset to $1$.
    When doing so it can turn at some point to state $\rightarrow_1$ or $\rightarrow_0$ depending on the layer of states $Q_t$ as detailed below: these states represent a head holding a YES/NO bit of information about the output of the test process happening on component $Q_t$.
    This bit must appear before reaching node $N$ and once appeared, this bit of information never changes in the future.
    When reaching node $N$ the head starts to move to the left, progressing in a background of Bs and letting symbols C behind (configuration types 4 and 5).
    At each move to the left, the second ${\{0,1\}}$ component of the node left by the head is reset to $1$.
    The fact that some ${\{0,1\}}$ component is reset to $1$ at each head move ensures that the corresponding configurations have more than one preimage (which is a key aspect when considering formula $\phi$).
    Finally, the head reaches node $1$ and must hold the output bit of the test process (configuration type 6), maintain it one step (configuration type 7), and finally erase it (type 8).
    Also, when reaching a configuration of type 6 at node $1$, the bit of the second ${\{0,1\}}$ component is reset to $1$ when the head at node $1$ is $\leftarrow_1$ and unchanged when it is $\leftarrow_0$.
    This bit is reset to $1$ in any case for configurations of type 7.
    As a result, a configuration of type 7 has exactly one preimage if and only if it is ${R_1C^{N-1}}$.
  \item The $Q_t$ component is the \emph{test} component, its role is to mark two positions in the configuration and interact with the head component in order to check a single V-constraint on the candidate configuration hold in the $Q$-component.
    More precisely, the test component ensures that the two marked positions are at distance $n$ (\textit{i.e.} they correspond to two vertical neighbors in the grid $(a_{i,j})$) and gathers locally at some node the information on the corresponding pair of states ${a_{i,j}a_{i,j+1}}$ and the constraint $V_{i,j}$ so that the head can check whether ${a_{i,j}a_{i,j+1}\in V_{i,j}}$. See Figure~\ref{fig:euclint} for an Euclidean intuition of how the distance equality test works.
    This behavior is implemented using alphabet ${Q_t = Q\times Q\times\{L,R,\alpha_L,\alpha_R,\beta,\gamma,\leftarrow_L,\leftarrow_R,\rightarrow_R\}}$ with freezing order: 
\[L<R<\alpha_L<\alpha_R<\leftarrow_L<\beta<\rightarrow_L<\leftarrow_R<\gamma.\]
    The third sub-component of $Q_t$ is used to mark two positions in the configuration as well as check that the distance between the two marked positions is exactly $n$ so that they indeed correspond to a pair of positions $(i,j)$ and $(i,j+1)$ in in the matrix ${(a_{i,j})}$.
    Its behavior is based on a set of valid configurations $\Sigma^+$ defined by local rules and synchronized with the $Q_h$ component.
    States $\leftarrow_L,\leftarrow_R,\rightarrow_R$ are called ``left/right arrows'' of the $Q_t$-component and are generated at specific positions when the global $Q_h$-head passes by (see Figure~\ref{fig:euclint}). 
    The two $Q$-sub-components of $Q_t$ are forced to hold states $a_{i,j}$ and $a_{i,j+1}$ respectively on valid configurations, and allow to check the V-constraint ${(a_{i,j},a_{i,j+1})\in V_{i,j}}$.
    The conditions defining $\Sigma^+$ are local (\textit{i.e.} they can be defined as a list of admissible pair of states between neighboring nodes) and a $\bot$ state is triggered whenever and wherever an invalid local pattern is detected.
    The conditions are the following:
    \begin{itemize}
    \item First, in the absence of a left-moving head in the $Q_h$ component, the two $Q$-sub-component must be uniform: each one is of the form $q^N$ for some $q\in Q$.
      When there is a left-moving head in the $Q_h$ component, each $Q$-sub-component is of the form: $q^iq_0^{N-i}$ where $q_0$ is the maximal state of $Q$ and $i$ is the position of the $Q_h$ head.
    \item Then, there are five types of admissible configurations on the third sub-component of $Q_t$:
      \begin{enumerate}
      \item ${L^*\alpha_L^+\alpha_R^+R^*}$ and $\alpha_L$ and $\alpha_R$ segments are forbidden to cross a block boundary (\textit{i.e.} node ${(n,j)}$ has an $\alpha_L$ if and only if ${(1,j+1)}$ has an $\alpha_R$),
      \item ${L^*\leftarrow_L\beta^*\alpha_L^*\alpha_R^*R^*}$,
      \item ${L^*\gamma^*\rightarrow_L\beta^*\alpha_L^*\alpha_R^*R^*}$ or ${L^*\gamma^*\rightarrow_L\beta^*\alpha_L^*\alpha_R^*\leftarrow_R\gamma^*R^*}$,
      \item ${L^*\gamma^*\rightarrow_L\leftarrow_R\gamma^*R^*}$,
      \item any configuration of the form ${c\gamma^*}$ where $c$ is the prefix of a configuration of type 3 or 4.
      \end{enumerate}
    \item Type 4 configurations are only authorized when $\rightarrow_L$ and $\leftarrow_R$ states meet at a bloc boundary, \textit{i.e.} are at positions of the form ${(n,j)}$ and ${(1,j+1)}$ (respectively).
    \item Moreover, only $L$, $\alpha_L$, $\alpha_R$ and $R$ are authorized in a node whose $Q_h$ component is in state $A$, therefore a type 1 configuration on the $Q_h$ component admits only a type $1$ configuration on the $Q_t$ component.
    \item Finally, in configurations of type $1$, let ${(i,j)}$ (\textit{i.e.} ${n(j-1)+i}$) be the leftmost node in state $\alpha_L$ and let $m$ be the rightmost node in state $\alpha_R$.
      Denote by $a$ and $b$ the states of the first and second $Q$-sub-component respectively.
      Then it must hold that $a$ is the state of the $Q$-component (the global one of the alphabet $Q'$) of node ${(i,j)}$ and $b$ is the state of the $Q$-component of node $m$.
    \end{itemize}
    The dynamics of this $Q_t$-component is as follows and respects the type order of configuration described above:
    \begin{itemize}
    \item Type 1 configurations don't change until the head of the $Q_h$-component arrives at node ${(i,j)}$ where it generates a $\leftarrow_L$ state.
    \item Then, $\leftarrow_L$ propagates in the $L$ background, letting $\beta$ states behind and until position ${(1,j)}$ is reached (\textit{i.e.} the first position to the left which is at the beginning of a bloc).
      Then, the arrow bounces by turning into $\rightarrow_L$ and starts to progress to the right letting $\gamma$ states behind.
    \item Meanwhile, when the $Q_h$ head reaches position $m$ (the rightmost node in state $\alpha$), it launches a $\leftarrow_R$ state in the $Q_t$ layer which starts to propagate to the left letting $\gamma$ sates behind.
    \item Also, when the $Q_h$ head bounces on node $N$ and starts to propagate to the left, it writes $q_0$ on each $Q$-sub-component and $\gamma$ on the third sub-component of $Q_t$, thus erasing progressively any information about the marked positions and the V-constraint being tested.
    \item The dynamics ends into the fixed point equal to $q_0^N$ on each $Q$-sub-component and $\gamma^N$ on the third sub-component.
    \end{itemize}
    Finally the $Q_t$-component influences the $Q_h$-component as follows: when the head of the $Q_h$-component of type $\rightarrow$ reaches node ${(i,j)}$ it becomes $\rightarrow_1$ if ${(a,b)\in V_{i,j}}$ (where $a$ and $b$ are the states of the $Q$-sub-components) and $\rightarrow_0$ else.
  \end{itemize}

  Let us now prove that this construction has the desired property.
  Let's call valid orbit any orbit without occurrence of $\bot$.
  \begin{claim}[$\phi$ checks V-constraints on valid orbits]\label{claim:coherenttest}
    Consider any valid orbit starting from a configuration $y$ without preimages, with ${y\to y^1}$ and ${\neg P_2(y^1)}$, and reaching a fixed point $x$.
    Then $y$ is of type 1 on components $Q_h$ and $Q_t$.
    Moreover, a correctly encoded test of V-constraint $V_{i,j}$ is encoded in component $Q_t$ and the configuration $z$ such that ${y\to^+z}$ and ${z\to x}$ verifies ${\neg P_2(z)}$ if and only if ${a_{i,j}a_{i,j+1}\in V_{i,j}}$.
  \end{claim}
  \begin{proof}
    Since there is no occurrence of $\bot$, the whole orbit belongs to $\Sigma^+$.
    A configuration of type 8 or 9 in the $Q_h$ component always has a preimage so $y$ is not of this type.
    A configuration of type 2,3,4,5,6 or 7 has a moving head that reset some $\{0,1\}$ component to $1$, so it cannot be the unique preimage of its successor, contradicting the hypothesis on $y^1$.
    Therefore $y$ is of type 1 on components $Q_h$ and $Q_t$.
    Then, by construction, the marked positions in the $Q_t$ component are at distance $n$ and there is a well-formed V-constraint test happening (otherwise a $\bot$ would be generated later in the orbit).
    The dynamics of the automata networks then ensures that the $Q_h$ heads holds the bit of information corresponding to the validity of the encoded V-constraint: it is $1$ if and only if ${a_{i,j}a_{i,j+1}\in V_{i,j}}$.
    The dynamics ends in a fixed point $x$ which has a configuration of type 8 on the $Q_h$-component.
    Already when reaching a configuration of type 6 or 7 or 8 on the $Q_h$-component, all the $Q_t$-component has been reset to a default value.
    Therefore it holds that the bit of information in the head is $1$ if and only if the type 7 configuration reached ${z=R_1C^{N-1}}$ has a unique preimage.\qed
  \end{proof}
  
  From the construction and Claim~\ref{claim:coherenttest} it should be clear that if the HV-domino CSP has a solution ${(a_{i,j})}$, then one can encode it into a fixed point configuration $x$ that satisfies the orbit property expressed in $\phi$ for all admissible choices of initial configuration $y$ (because all admissible V-constraint tests are satisfied by the CSP solution).
  In this case the automata network verifies $\phi$.

  Conversely, if the automata network verifies $\phi$ and if the fixed point $x$ can be chosen to be a configuration without $\bot$, then this configuration encodes a solution to the HV-domino CSP by Claim~\ref{claim:coherenttest} and because any valid V-constraint test can be encoded in an appropriate initial configuration $y$.
  It remains to discard the possibility that $\phi$ is valid because $x$ is chosen to be the invalid fixed-point $\bot^N$, this is the purpose of the following claim.
  
  \begin{claim}[$\phi$ discards invalid orbits]\label{claim:errorpreimages}
    Consider three configurations $y,z,x$ such that ${y\to^+z\to x}$ and ${x\to x}$ and $x\neq z$. If ${\neg P_2(z)}$ then $x$ cannot be the configuration $\bot^N$.
  \end{claim}
  \begin{proof}
    First $z$ must have an occurrence of $\bot$ because it is impossible that the preimage $z'$ of $z$ be everywhere correct and in one step becomes a configuration $z$ everywhere incorrect but without occurrence of $\bot$: indeed, by construction, the changes not involving $\bot$ state that can occur in a configuration in one step are only in the neighborhood of arrow states of both $Q_h$ and $Q_t$ components, and they have a bounded number of occurrences by definition of $\Sigma^+$.
    Moreover, there must be an occurrence of $\bot$ in $z$ at position $i$ such that ${z'(i)\neq\bot}$.
    Indeed, otherwise it would imply ${z=\bot^N}$ which is impossible under the hypothesis.
    Therefore by just changing the dummy component at $i$ in $z'$ we produce another preimage of $z$, so $P_2(z)$ holds which is a contradiction.\qed
  \end{proof}
  We have thus shown that the HV-domino CSP has a solution if and only if the automata network verifies $\phi$.
  The theorem follows since the construction can be computed efficiently (actually in LOGSPACE).\qed
\end{proof}

\section{Conclusion}

 Our results contribute to the following global picture about computability and complexity, comparing both finite automata network versus infinite CAs, and the freezing case versus the general case.
 Each cell of the table is divided between the general case (lower left in black) and the freezing case (upper right in blue).

\newcommand{\generalfreezing}[2]{{\diagbox[width=.4\textwidth]{#1}{\color{blue}#2}}}
\newcommand{\resref}[2]{
  \begin{tabular}{c}
    #1\\
    {\tiny #2}
  \end{tabular}
}

\begin{center}
\begin{tabular}{c|c|c}
  & Infinite 1D CA & Finite bounded pathwidth AN \\
  \hline
  Nilpotency & \generalfreezing{Undecidable \cite{kari92}}{\resref{Decidable}{\cite[Theorem 2]{dmtcs:9004}}} & \generalfreezing{\resref{PSPACE-complete}{\cite[Corollary 3.2]{stacs.2021.32}}}{\resref{co-NL}{(Corollary~\ref{cor:nilconl})}} \\
  \hline
  \resref{Regular trace}{properties} & \generalfreezing{Undecidable}{\resref{Undecidable}{\cite[Theorem 5]{dmtcs:9004}}} & \generalfreezing{\resref{PSPACE-complete}{\cite[Theorem 3.3]{stacs.2021.32}}}{\resref{NL-complete}{(Theorem~\ref{theo:nlhard})}}  \\
  \hline
  \fop & \generalfreezing{Undecidable \cite{kari92}}{Open} & \generalfreezing{\resref{PSPACE-complete}{\cite[Corollary 3.2]{stacs.2021.32}}}{\resref{NP-hard}{(Theorem~\ref{theo:fop-hardness})}} 
\end{tabular}
\end{center}

The obvious continuation of our work would be to study model checking of \fop logic for one-dimensional freezing cellular automata. We conjecture that there exists a fixed formula ${\phi\in FO^+}$ such that determining whether a given freezing CA has property $\phi$ is undecidable. The table above recall that such a property $\phi$, if it exists, cannot be equivalent to the nilpotency property.

\section{Declarations}

\noindent\textbf{Ethical statements.} \textit{Not applicable.}\\

\noindent\textbf{Competing interest.}  \textit{Not applicable.}\\

\noindent\textbf{Authors' contributions.} \textit{Contribution to be considered equal among all authors, alphabetical order used.}\\

\noindent\textbf{Funding.} \textit{Research partially supported by projects STIUC-AMSUD 22-STIC-02 (all authors), Fondecyt-ANID 1200006 (EG), FONDECYT-ANID 1230599 (PM), ANID FONDECYT Postdoctorado 3220205 (MRW).}\\

\bibliographystyle{plain}
\bibliography{paper}

\end{document}